\documentclass[12pt,transaction, draftclsnofoot,onecolumn]{IEEEtran}
\usepackage{amsfonts}
\usepackage{amssymb}
\usepackage{mathtools}
\usepackage{hyperref}
\usepackage{graphicx}
\usepackage{subfigure}
\usepackage{enumerate}
\usepackage{amsmath}
\usepackage{color}
\usepackage{amsthm}
\usepackage{algorithm}
\usepackage{algpseudocode}
\usepackage{accents}
\usepackage{stfloats}
\theoremstyle{definition}
\DeclareSymbolFont{largesymbolsA}{U}{txexa}{m}{n}
\DeclareMathSymbol{\varprod}{\mathop}{largesymbolsA}{16}
\algtext*{EndWhile}% Remove "end while" text
\algtext*{EndIf}% Remove "end if" text
\algtext*{EndFor}% Remove "end for" text

\newcommand\ubar[1]{\underaccent{\bar}{#1}}

\DeclareMathOperator*{\argmin}{arg\,min}
\DeclareMathOperator*{\argmax}{arg\,max} % Jan Hlavacek

\newcommand{\bp}{\begin{proof} \small }
\newcommand{\ep}{\end{proof} \normalsize}
\newcommand{\epx}{\end{proof} \small}
\newcommand{\bpa}{\begin{proofappx} \footnotesize }
\newcommand{\epa}{\end{proofappx} \small }
\newtheorem{theorem}{Theorem}

\newtheorem{lemma}{Lemma}
\newtheorem{assumption}{Assumption}
\newtheorem{definition}{Definition}

\newtheorem*{theorem*}{Theorem}
\newtheorem*{proposition*}{Proposition}
\newtheorem*{corollary*}{Corollary}
\newtheorem*{lemma*}{Lemma}
\newtheorem*{assumption*}{Assumption}
\newtheorem*{definition*}{Definition}
\newtheorem*{claim*}{Claim}

\newcommand{\bm}[1]{\mbox{\boldmath $#1$}}

\newcommand{\be}{\begin{equation}}
\newcommand{\ee}{\end{equation}}
\newcommand{\bs}{\begin{subequations}}
\newcommand{\es}{\end{subequations}}
\newcommand{\bq}{\begin{eqnarray}}
\newcommand{\eq}{\end{eqnarray}}
\newcommand{\bqn}{\begin{eqnarray*}}
\newcommand{\eqn}{\end{eqnarray*}}

\newcommand{\ba}{\left[ \begin{array}}
\newcommand{\ea}{\\ \end{array} \right]}
\newcommand{\ben}{\begin{enumerate}}
\newcommand{\een}{\end{enumerate}}

\def\f{{\boldsymbol{f}}}

\def\p{{\boldsymbol{p}}}

\def\x{{\boldsymbol{x}}}

\def\real{{\mathchoice%
{\hbox{\rm\setbox1=\hbox{I}\copy1\kern-.45\wd1 R}}
{\hbox{\rm\setbox1=\hbox{I}\copy1\kern-.45\wd1 R}}
{\hbox{\scriptsize\rm\setbox1=\hbox{I}\copy1\kern-.45\wd1 R}}
{\hbox{\scriptsize\rm\setbox1=\hbox{I}\copy1\kern-.45\wd1 R}}}}

\def\Zint{{\mathchoice{\setbox1=\hbox{\sf Z}\copy1\kern-.75\wd1\box1}
{\setbox1=\hbox{\sf Z}\copy1\kern-.75\wd1\box1}
{\setbox1=\hbox{\scriptsize\sf Z}\copy1\kern-.75\wd1\box1}
{\setbox1=\hbox{\scriptsize\sf Z}\copy1\kern-.75\wd1\box1}}}
\newcommand{\complex}{ \hbox{\rm C\kern-0.45em\rule[.07em]{.02em}{.58em}%
\kern 0.43em}}

\makeatletter
\newcommand{\algmargin}{\the\ALG@thistlm}
\makeatother
\newlength{\whilewidth}
\algdef{SE}[parWHILE]{parWhile}{EndparWhile}[1]
{\parbox[t]{\dimexpr\linewidth-\algmargin}{%
		\hangindent\whilewidth\strut\algorithmicwhile\ #1\ \algorithmicdo\strut}}{\algorithmicend\ \algorithmicwhile}%
\algnewcommand{\parState}[1]{\State%
	\parbox[t]{\dimexpr\linewidth-\algmargin}{\strut #1\strut}}

\ifodd 1

\else

\fi

\begin{document}
\title{Budget-constrained Edge Service Provisioning with Demand Estimation via Bandit Learning}
\author{Lixing Chen ~\IEEEmembership{Student Member,~IEEE}
        ~~Jie Xu~\IEEEmembership{Member,~IEEE}% <-this % stops a space
\thanks{L. Chen and J. Xu are with the Department of Electrical and
	Computer Engineering, University of Miami. Email: \{lx.chen, jiexu\}@miami.edu}}

\maketitle
\vspace{-0.4 in}
\begin{abstract}
Shared edge computing platforms, which enable Application Service Providers (ASPs) to deploy applications in close proximity to mobile users are providing ultra-low latency and location-awareness to a rich portfolio of services. Though ubiquitous edge service provisioning,  i.e., deploying the application at all possible edge sites, is always preferable, it is impractical due to often limited operational budget of ASPs. In this case, an ASP has to cautiously decide where to deploy the edge service and how much budget it is willing to use. A central issue here is that the service demand received by each edge site, which is the key factor of deploying benefit, is unknown to ASPs a priori. What's more complicated is that this demand pattern varies temporally and spatially across geographically distributed edge sites. In this paper, we investigate an edge resource rental problem where the ASP learns service demand patterns for individual edge sites while renting computation resource at these sites to host its applications for edge service provisioning. An online algorithm, called Context-aware Online Edge Resource Rental (COERR), is proposed based on the framework of Contextual Combinatorial Multi-armed Bandit (CC-MAB). COERR observes side-information (context) to learn the demand patterns of edge sites and decides rental decisions (including where to rent the computation resource and how much to rent) to maximize ASP's utility given a limited budget. COERR provides a provable performance achieving sublinear regret compared to an Oracle algorithm that knows exactly the expected service demand of edge sites. Experiments are carried out on a real-world dataset and the results show that COERR significantly outperforms other benchmarks.

\end{abstract}

%\begin{IEEEkeywords}
%IEEEtran, journal, \LaTeX, paper, template.
%\end{IEEEkeywords}

% \fi\left( 
\IEEEpeerreviewmaketitle

\section{Introduction}\label{sec:introduction}
The prevalence of ubiquitously connected smart devices and the Internet of Things are driving the development of intelligent applications, turning data and information into actions that create new capabilities, richer experiences, and unprecedented opportunities. As these applications become increasingly powerful, they are also turning to be more computational-demanding, making it difficult for resource-constrained mobile devices to fully realize their functionalities solely. Although mobile cloud computing \cite{dinh2013survey} has provided mobile users with a convenient access to a centralized pool of configurable and powerful computing resources, it is not a ``one-size-fit-all'' solution due to the stringent latency requirement of emerging applications and often unpredictable network condition. In addition, as the mobile applications (e.g. mobile gaming and virtual/augmented reality) are becoming more data-hungry, it would be laborious to transmit all these data over today's already congested backbone network to the remote cloud. 

As a remedy, Mobile Edge Computing (MEC) \cite{satyanarayanan2017emergence} has been recently proposed as a new computing paradigm to enable service provisioning in close proximity of user devices at the network edge, thereby enabling analytics and knowledge generation to occur closer to the data source and providing low-latency responses. Such an edge service provisioning scenario is no longer a mere version but becoming a reality. For example, Vapor IO's Kinetic Edge \cite{vaporio} places edge data centers at the base of cell towers and nearby aggregation hubs, thereby bringing cloud-like services to the edge of the wireless network. Kinetic Edge has started in Chicago and is rapidly expanding to the other US cities. It is anticipated that cloud providers, web scale companies, and other enterprises will soon be able to rent computation resources at these shared edge computing platforms to deliver edge applications in a flexible and economical way without building their own data center or trenching their own fiber.

However, how to effectively and efficiently deliver edge service in such a shared edge system faces many special issues. Firstly, the benefit of deploying application service at a certain edge server mainly depends on the number of edge task requests received from the users, yet this service demand is unknown to the Application Service Provider (ASP) before deploying applications at edge servers. What's more complicated is that the service demand is uncertain in both temporal and spatial domains, i.e., the demand pattern of an edge site varies across the time and the demand patterns at geographically distributed edge sites may not replicate a global demand pattern. How to learn the service demand pattern for each edge site precisely with \emph{cold-start} (i.e., no prior knowledge available) is the very first step toward efficient edge service provisioning. Secondly, to deploy services at the edge sites, ASP needs to rent a certain amount of computation resource to host its applications. While renting a sufficient amount of computation resource at every possible edge site can deliver the best Quality of Service (QoS), it is practically infeasible especially for small and starting ASPs due to the prohibitive budget requirement. In common business practice, an ASP has a budget on the operating expenses and desires the best performance within the budget \cite{stan2018uptime}. This means that the ASP can only rent limited computation resource at a limited number of edge sites and hence, which edge sites to deploy applications and how much computation resource to rent at these sites must be judiciously decided to optimize QoS given the limited budget. Thirdly, the service demand estimation and edge resource rental are not two independent problems but closely intertwined during online decision making. On the one hand, renting computing resource and deploying application service at an edge site allows the ASP to collect historical data on the received service demand for better demand estimation. On the other hand, accurate demand estimations help the ASP optimize its computation resource rental and achieve a higher QoS. Therefore, an appropriate balance should be made between these two purposes to maximize the utility of ASP in the long run. The main contributions of this paper are summarized as follows: \\
1) We formulate an edge resource rental (ERR) problem where ASP rents computation resource at edge servers to host its applications for edge service provisioning. ERR is a three-fold problem in which the ASP needs to (i) estimate the service demand received by edge sites with cold-start, (ii) decide whether edge service should be provided at a certain edge site, and (iii) optimize how much resource to rent at edge sites to maximize ASP's utility under a limited budget.   \\
2) An online decision-making algorithm called \emph{Context-aware Online Edge Resource Rental} (COERR) is proposed to solve the ERR problem. COERR is designed in the framework of Contextual Combinatorial Multi-armed Bandit (CC-MAB). The ``contextual'' nature of COERR allows the ASP to observe the side-information (context) of edge sites for the service demand estimation and the ``combinatorial'' nature of COERR enables the ASP to rent computation resource at multiple edge servers for utility maximization. \\
3) We analytically bound the performance loss, termed \emph{regret}, of COERR compared to an Oracle benchmark that knows the expected service demand of each edge site. The regret bound is first given in a general form available for \textit{arbitrary} estimators and algorithm parameters. A specific sublinear regret upper bound is then derived in a concrete setting by specifying the applied service demand estimators and algorithm parameters, which not only implies that COERR produces asymptotically optimal rental decisions but also provides finite-time performance guarantee. \\
4) We carry out extensive simulations using the real-world service demand trace in Grid Workloads Archive (GWA) \cite{iosup2008grid}. The results show that the proposed COERR algorithm significantly outperforms other benchmark algorithms.        

The rest of this paper is organized as follows: Section \ref{sec:related_work} reviews related works. Section \ref{sec:system_model} presents the system model for edge resource rental problem. Section \ref{sec:coerr} designs the context-aware online edge resource rental (COERR) algorithm and provides analytical performance guarantee. Section \ref{sec:ext} discusses the extension of COERR when applied with approximated solutions for per-slot utility maximization. Section \ref{sec:experiment} shows the experiment results of the proposed algorithm on a real-world service demand trace, followed by the conclusion in Section \ref{sec:conclusion}. 

\section{Related Work}\label{sec:related_work}
Driven by the promising properties and tempting business opportunities, Mobile Edge Computing (MEC) \cite{satyanarayanan2017emergence,mao2017mobile} is attracting more and more attention from both academia and industry. Various works have studied from different aspects of MEC, including edge platform design \cite{saguna} for integrating edge computing platform into edge facilities (e.g., Radio Access Network \cite{saguna}), computation offloading \cite{chen2017socially} for deciding what/when/how to offload tasks from user's mobile devices to edge servers, and edge orchestration \cite{liu2018orchestrator} for coordinating the distributed edge servers.  However, these edge computing topics all rest on the assumption that the computing resources and capabilities have been provisioned to ASP at the edge sites. By contrast, this paper focuses on the problem that how should ASP rent computation resource and place edge applications among many possible edge sites such that the users can better enjoy the edge service access.

Service placement in shared edge systems has been studied in many contexts in the past. Considering content delivery as a service, many prior works study caching content replicas in traditional content delivery networks (CDNs) \cite{chen2002dynamic} and, more recently, in wireless caching systems such as small cell networks \cite{shanmugam2013femtocaching}. However, content caching concerns only data content caching given storage constraints at edge facilities while placing edge applications needs to take into account the computation resources at the edge servers. Service placement for MEC is recently studied in \cite{zhao2018red}, where the authors consider a hierarchical edge-cloud system and an online replacement policy is designed to minimize the cost of forwarding requests to Cloud and downloading new services to edge server. While \cite{zhao2018red} uses a placing-upon-request model, our work is in a proactive manner where the deploy applications at the beginning of each decision cycle based on service demand estimation. The authors in \cite{xu2018joint,xie2018dynamic} investigate service placement/caching to improve the efficiency of edge resource utilization by enabling cooperation among edge servers. However, these works assume that the service demand is known a priori whereas the service demand pattern in our problem has to be learned over time. A learning-based edge service placement is proposed in \cite{chen2018spatio}, which uses bandit learning similar to the framework in this paper. However, it only addresses the problem of where the service should be placed but does not optimize how much computation resource needs to rent at edge servers. However, in practice, an ASP has to decide the amount of computing resource to rent when placing the service application. This paper helps the ASP to determine the amount of computing resource to rent at edge servers when placing the edge service. In addition, \cite{chen2018spatio} uses a simple sample mean to estimate the service demand, but we generalize our algorithm to work with an arbitrary estimator.

MAB algorithms have been studied to address the tradeoff between exploration and exploitation in sequential decision making under uncertainty \cite{lai1985asymptotically}. The classic MAB algorithm, e.g. UCB1, concerns with learning the single optimal action among a set of candidate actions with unknown rewards by sequentially trying one action each time and observing its realized noisy reward \cite{auer2002finite}. Combinatorial bandits extends the basic MAB by allowing multiple-play each time (i.e. renting computation resources at multiple edge servers under a budget in our problem) \cite{gai2012combinatorial} and contextual bandits extends the basic MAB by considering the context-dependent reward functions \cite{li2010contextual,tekin2015distributed}. While both combinatorial bandits and contextual bandits problems are already much more difficult than the basic MAB problem, this paper tackles the even more difficult CC-MAB problem. Recently, a few other works \cite{li2016contextual,muller2017context,chen2018contextual} also started to study CC-MAB problems. However, these works make assumptions that are not suitable for our problem. \cite{li2016contextual} assume that the reward of an individual action is a linear function of the contexts, which is less likely to be true in practice. In \cite{muller2017context}, the exact solution to its per-slot problems can be easily derived, however, in our problem, the per-slot problem is a Knapsack problem with conflict graph (KCG) whose optimal solution cannot be efficiently derived and hence, we also investigate the impact of approximation solution on CC-MAB framework. Though \cite{chen2018contextual} also considers the approximation solutions, it is given for a special case (greedy algorithm for submodular function maximization). The key difference of our CC-MAB is that it does not simply decide which arms to pull (i.e., which edge sites to place applications), it also chooses the configuration of arms (i.e., how much resource to rent at each edge site).

\section{System Model} \label{sec:system_model}
\subsection{Network Structure and Resource Rental}
We consider a typical scenario where edge computing is enabled in a  heterogeneous small-cell network as illustrated in Fig.\ref{fig:illu_system_model}. The small cell network consists of a set of Small Cells (SCs), indexed by $\mathcal{N} = \{1,2,\dots,N\}$, and a macro base station (MBS), indexed by $0$. Each SC has a small-cell base station (SBS) equipped with a shared edge computing platform. The edge platforms use virtualization technology for flexible allocation of computation resource, e.g., CPU frequency and RAM. ASPs sign contracts with SBSs to rent computation resource at co-located edge servers in order to host their application and provide service access to subscribed users. SBSs provide Software-as-a-Service (SaaS) to ASPs, managing computation resources requested by ASPs using virtualization, while the ASPs maintain its own user data serving as a middleman between end users and SaaS Providers. As such, SBSs charge ASPs for the amount of requested computation resource. Besides the SBSs and edge servers, there also exists an MBS that provides ubiquitous radio coverage and is connected to the cloud server in case that edge service is not available for users.
\begin{figure*}[tb]
	\centering
	\includegraphics[width=0.6\linewidth]{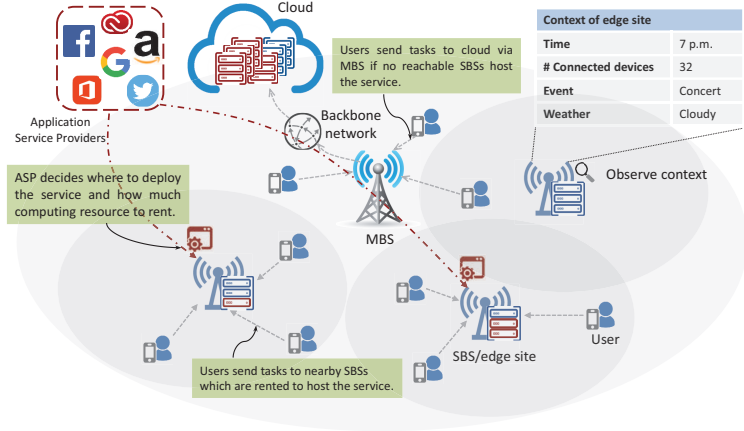}
	\vspace{-0.15 in}
	\caption{Heterogeneous small-cell network for edge computing. The red machines in edge servers or Cloud denote the computation resource rented by the ASP.}
	\label{fig:illu_system_model}
	\vspace{-0.2 in}
\end{figure*}
The contract of edge resource rental is signed for a fix length of time span (e.g., 3 hours or half a day). Therefore, we discretize the operational timeline into time slots. At the beginning of each time slot, ASP determines the amount of computation resource to rent from SBSs for application service deployment. In particular, we consider the processor capacity (i.e. CPU frequency) as the key component of computing resource since it decides the processing delay of tasks at edge servers as considered in most existing works \cite{chen2018computation, mao2016dynamic}. The other resource components, e.g., RAM, storage, I/O, are matched to the rented processor capacity. Let $f^t_n\in \left\{0\cup [f^\text{min}_n, f^\text{max}_n]\right\}$ denote the processor capacity rented by the ASP at SBS $n$ in time slot $t$, where $f^\text{min}_n$ is the minimal rental contract (i.e. the ASP have to at least rent $f^\text{min}_n$ to set up a virtualized computing platform at SBS $n$) and $f^\text{max}_n$ is the maximum computation resource that can be rented by the ASP at SBS $n$. Based on the state-of-the-art virtualization technologies, the resource allocation at edge server is often realized using \emph{containerization} or \emph{virtual machine} and hence we assume that each SBS $n$ discretized their computation resource into containers or VMs. In this case, the feasible rental decisions at SBS $n$ can be collected in a rental decision set $F_n$. Let $\f^t = \{f^t_1,f^t_2,\dots,f^t_N\}, f^t_n\in F_n$ be the computation resource rented by ASP at all SBSs. The vector $\f^t \in \mathcal{F} \triangleq \varprod_{n\in\mathcal{N}}F_n$ is referred as ASP's rental decision in $t$.

Each SBS sets a price for its computation resource. Let $w_n(f_n)$ denote the price charged by SBS $n$ if ASP rents $f_n$ processor capacity at SBS $n$, where $w_n(\cdot)$ is a non-decreasing mapping function\footnote{The price mapping could be non-decreasing linear/nonlinear functions or tables and each SBS may have its own mapping.} that determines the price for computation resource $f_n$. Due to the limited budget, the resource rental decision of ASP must satisfy the budget constraint $\sum_{n=1}^N w_n(f^t_n) \leq B, \forall t$, where $B$ is ASP's budget. Note that $f^\text{min}_n$, $f^\text{max}_n$, $F_n$, and $w_n(\cdot)$ may possibly vary across the time slots due to certain auction strategies or stochastic resource scheduling policies carried out by SBSs. To keep the system model simple, we assume that these parameters are constants. However, the proposed method is also compatible with time-varying system parameters. In addition, we consider edge resource rental problem for one ASP in this paper, the edge system may need strategies, e.g., \emph{first-come-first-served} or matching algorithms \cite{sonmez2011matching} to coordinate multiple ASPs.   

Besides the edge servers, ASP also possesses an entrepreneur cloud or a configured platform at the commercial cloud to provide ubiquitous application service. The processor capacity of the cloud service is denoted by $f^t_0$. Usually, we will have $f^t_0\gg f^t_n$. 

\subsection{Service Delay for Edge Computing}
During a time slot, users in the edge system have computation tasks to be offloaded to edge/cloud servers for processing. We assume that the input data size of one task is $s$ in bits and the number of required CPU cycles to process one task is $c$. If a user device is covered by an SBS, it can offload computational tasks to the edge server co-located with the SBS. The service delays are incurred for completing these tasks using edge computing; it consists of two main parts: transmission delay and processing delay. 
\subsubsection{Transmission delay} User's tasks are sent via the one-hop wireless connection to SBSs. Note that the time scale of edge resource rental cycles (e.g., half a day) is much larger than that of task offloading cycle (few seconds), an SBS may receive a large number of tasks, indexed by $k=1,2,\dots,K$, in time slot $t$. For each task $k$, the uplink transmission rate can be calculated by Shannon Capacity: 
\begin{align}
	r_k = W_k \log\left(1+\frac{P_k H_k}{I^\text{inter-cell}_k + I^\text{intra-cell}_k +\sigma^2}\right),
\end{align}
where $W_k$ is the allocated bandwidth, $P_k$ is the transmission power of the user, $H_k$ is the channel gain, $I^\text{inter-cell}_k$ and $I^\text{intra-cell}_k$ are the inter/intra-cell interferences, and $\sigma^2$ is the noise power. It is difficult to know exactly the data rate for transmitting each task during the planning stage due to unpredictable interference, fading, etc. Instead of considering the transmission rate for each task, we operate SBSs to work on an expected transmission rate $r^t$ in each time slot $t$, i.e., we expect $\frac{1}{K}\sum_{K}^{k=1} r_k$ equals $r^t$. Such an requirement on the expected transmission rate can be satisfied by state-of-the-art spectrum allocation method \cite{chandrasekhar2009spectrum}. We denote by $r^t_n$ the expected uplink transmission rate of SBS $n$ in time slot $t$. Then, the expected transmission delay for a user to transmit one task to SBS $n$ is $d^{\text{tx},t}_n = s/r^t_n$. To keep the system model simple, we assume the data size of task result is small and therefore its downloading time can be neglected. However, adding result downloading time does not make a big difference and our algorithm can still be applied.

\subsubsection{Processing delay} The processing delay of edge computing is determined by the processor capacities rented by the ASP at SBSs. We assume that the edge server admits at most $\lambda^{\max}$ tasks in a time slot to avoid overloading and queuing delays at edge servers. Given the processor capacity $f^t_n > 0$, the processing delay for one task at SBS $n$ can be obtained as: $d^{\text{proc},t}_n (f^t_n)= \frac{c}{f^t_n}$. Therefore, the service delay for one task at SBS $n$ is: 
\begin{align}
d^t_n(f^t_n) = d^{\text{tx},t}_n  + d^{\text{proc},t}_n (f^t_n)= \frac{s}{r^t_n}+ \frac{c}{f^t_n}, ~n= 1,2,\dots,N.
\end{align}

\subsection{Service Delay for Cloud Computing}
If a user has no accessible SBSs or its task request is rejected by an edge server due to overloading, it then has to offload its tasks to Cloud via an MBS. Similarly, the service delay for Cloud computing also consists of transmission delay and processing delay.
\subsubsection{Transmission delay}
Besides the wireless transmission delay incurred by sending the tasks from users to MBS, the offloaded tasks have to travel through congested backbone Internet, which incurs large backbone transmission delay, to reach the remote cloud server. We assume, similar to SBSs, that the MBS applies state-of-the-art channel/power/interference management strategies to guarantee an expected wireless transmission rate $r^t_0$ in time slot $t$. Therefore, the expected wireless transmission delay for one task is $s/r^t_0$. The backbone transmission delay is mainly determined by the backbone transmission rate, which is a random variable based on the network condition. Let $v^t$ be the expected backbone transmission delay and $h^t$ be the round-trip time in time slot $t$, then the expected backbone transmission delay for one task can be obtained as $s/v^t + h^t$. Taking into account all the above components, the expected transmission delay for one task using cloud computing can be obtained by: $d^{\text{tx},t}_0 = \frac{s}{r^t_0} + \frac{s}{v^t} + h^t$.

\subsubsection{Processing delay}
Since the cloud server has unlimited computation resources, we assume that the cloud server has no admission constraints. Recall that the processor capacity allocated for each task at ASP cloud is $f^t_0$, the processing delay for one task using cloud computing can be expressed easily as $d^{\text{proc},t}_0 = c/f^t_0$. 

The expected service delay for one task using cloud computing is therefore: 
\begin{align}
d^t_0 = d^{\text{tx},t}_0  + d^{\text{proc},t}_0 = \frac{s}{r^t_0} + \frac{s}{v^t} +\frac{c}{f^t_0} + h^t.
\end{align}

We assume that the maximum service delay for one task is bounded, i.e., $d^t_0, d^t_n \leq d^{\max}$. This is a practical assumption in edge computing since if the service delay of edge/cloud computing is too large the mobile devices can always choose to process the tasks locally, which guarantees a service delay $d^{\max}$.

\subsection{ASP Utility Function}
The applications deployed at the network edge improve QoS for users by providing low-latency response. The ASP derives utilities from the improved QoS, which is defined as delay reduction achieved by deploying services at edge servers. Let 
\begin{equation}
\Delta^t_n(f^t_n) = \left\{
\begin{split}
&d^t_0-d^t_n(f^t_n), & f_n>0\\
&0,& f_n =0
\end{split}\right.
\end{equation}
be the delay reduction of a task processed by SBS $n$ instead of Cloud and let $\lambda^t_n$ be the service demand within the coverage of SBS $n$. Note that $\lambda^t_n$ does not equal the service demand received by SBS $n$ since task requests will be offloaded to the cloud server if the ASP rents no computation resource at SBS $n$. Therefore, the total utility achieved by SBS $n$ is:
\begin{align}
	u^t_n(f^t_n;\lambda^t_n) = \min\{\lambda^t_n,\lambda^{\max}(f^t_n)\} \cdot \Delta^t_n(f^t_n).
\end{align}
where $\lambda^{\max}(f^t_n)$ is the maximum service demand can be processed by an SBS depending on the amount of rented computing resource $f^t_n$. Intuitively, more tasks can be process at an SBS when more computing resource are rented. Therefore, the function $\lambda^{\max}(\cdot)$ should be non-decreasing on $f^t_n$. Notice that the service delay for a task is bounded by $d^t_0, d^t_n \leq d^{\max}$, we have $\Delta^t_n(f^t_n)\leq d^{\max}$. The total ASP utility is 
\begin{align}
U^t(\f^t; \bm{\lambda}^t) = \sum\nolimits_{n\in\mathcal{N}} u^t_n(f^t_n;\lambda^t_n)
\end{align}
where $\bm{\lambda}^t = \{\lambda_1,\lambda_2,\dots,\lambda_N\}$ collects the service demands within the coverage of all $N$ SBSs.
\subsection{Problem Formulation}
The edge resource rental (ERR) problem for ASP is a sequential decision-making problem. The goal of ASP is to make rental decision $\f^t, \forall t$ to maximize the expected utility up to time horizon $T$. Since the service demand $\bm{\lambda}^t, \forall t$ of SBSs is not known to the ASP when making its rental decision, we write it as $\hat{\bm{\lambda}}^t$ that needs to be estimated at the beginning of each time slot. Therefore, the edge resource rental problem can be written as:

\begin{subequations}
	\begin{align}
	\textbf{ERR}:~~	\max_{\{\f^t\}_{t=1}^T} &~\sum\nolimits_{t=1}^T U^t(\f^t; \hat{\bm{\lambda}}^t) \\
		\text{s.t.} ~~~
		& f^t_n \in \left\{0 \cup [f_n^{\min}, f_n^{\max}]\right\}, \forall n\in \mathcal{N}, \forall t\\
		& f^t_n \in F_n, \forall n\in \mathcal{N}, \forall t\\
		& \sum\nolimits_{n=1}^N w_n(f^t_n) \leq B , \forall t
	\end{align}
\end{subequations}

There are several challenges to be addressed and should be addressed simultaneously to solve the ERR problem: (i) One of the key challenges of ERR is to make precise service demand estimation, such that the derived rental decision is able to produce the expected utility when implemented. Since the algorithm is run with cold start, the algorithm should also collect the historical data for making estimations. Note that the service demand received by an SBS is revealed to ASP only when the application is deployed ($f_n>0$) at the SBS. Though the service demand received by the cloud server can also be observed, it does not help much to learn the service demand of a specific SBS due to the fact that the location information of users is usually veiled to ASP due to the privacy concerns. Therefore, the rental decision making should take into account the data collection for demand estimation. (ii) With the service demand estimations, how to optimally determine the rental decision at each SBSs given the limited budget should be carefully considered. (iii) Since the rental decisions are made based on the estimated service demand, the accuracy of demand estimation will have a deterministic impact on ASP's utility. The ASP needs to decide when the estimation is accurate enough for guiding the computation resource rental and when more data should be collected to produce a better demand estimation. In the next section, we propose an algorithm based on the multi-armed bandit framework to address the mentioned issues at the same time.

\section{Edge Resource Rental as Contextual Combinatorial Multi-Armed Bandits}\label{sec:coerr}
In this section, we formulate our ERR problem as a Contextual Combinatorial Multi-Armed Bandit (CC-MAB). The problem is ``combinatorial'' because ASP will rent computation resource at multiple SBSs under a budget constraint. The problem is ``contextual'' because we will utilize context associated with SBSs to infer their service demand. In general, the contextual bandit is more applicable than non-contextual variants as it is rare that no context is available \cite{langford2008epoch}. In our problem, the service demand received by an SBS depends on many factors, which are collectively referred to as \emph{context}. For example, the relevant factors can be the user factor (e.g. user population, user type), temporal factor (e.g., time in a day, season), and external environment factors (e.g., events such as concerts). This categorization is clearly not exhaustive and the impact of each single context dimension on the service demand is unknown. Our algorithm learns to discover the underlying connection between context and service demand pattern over time. 

In CC-MAB, ASP observes the context of SBSs at the beginning of each time slot before making the rental decision. Let $x^t_n \in \mathcal{X}$ be the context of SBS $n$ observed in time slot $t$, where $\mathcal{X}$ is the context space. Without loss of generality, we assume that the context space is bounded and hence can be denoted as $\mathcal{X} = [0,1]^D$, $D$ is the number of context dimension. The context of all SBSs are collected in $\x^t = \{x^t_1,x^t_2,\dots,x^t_N\}$. The service demand $\lambda^t_n$ received by SBS $n$ is a random variable parameterized by the context $x^t_n$. Let $\lambda_n: \mathcal{X} \to \Lambda_n$ be the mapping that maps a context $x^t_n \in \mathcal{X}$ to SBS $n$'s service demand distribution $\lambda_n(x^t_n)$. We rewrite the service demand vector in a context-aware form: $\bm{\lambda}^t = \{\lambda_1(x^t_1), \lambda_2(x^t_2),\dots, \lambda_N(x^t_N) \}$. In addition, we let $\mu_n(x^t_n)\triangleq\mathbb{E}[\lambda_n(x^t_n)]$ be the expected value of the service demand distribution $\lambda_n(x^t_n)$. The vector $\bm{\mu}^t = \{\mu_1(x^t_1), \mu_2(x^t_2),\dots,\mu_N(x^t_N)\}$ collects the expected service demands for all SBSs given context $\x^t$.

\subsection{Oracle Solution and Regret}
Before proceeding with the algorithm design, we first give an Oracle benchmark solution to the ERR problem by assuming that the ASP knows exactly the context-aware service demand $\mu_n(x), \forall x\in\mathcal{X}$. In such a case, the ERR problem can be decoupled into $T$ independent subproblems, one for each time slot $t$, as below:
\begin{subequations}\label{opt:subproblem}
	\begin{align}
	\textbf{Sub-problem}:~~&\max_{\f^t} ~ U^t(\f^t; \bm{\mu}^t) \\
	\text{s.t.} ~~~
	& f^t_n \in \{0 \cup [f_n^{\min}, f_n^{\max}]\}, \forall n\in\mathcal{N} \label{cstr:sub_c1}\\
	& f^t_n \in F_n, \forall n\in\mathcal{N} \label{cstr:sub_c2}\\
	& \sum\nolimits_{n\in\mathcal{N}} w_n(f^t_n) \leq B \label{cstr:sub_c3}
	\end{align}
\end{subequations}
where the service demand estimation $\hat{\bm{\lambda}}^t$ is replaced by $\bm{\mu}^t$. The above subproblem is an combinatorial optimization problem with Knapsack constraints. The optimal solution to each subproblem can be derived by \emph{brute-force} if the size of action space $\mathcal{F}$ is moderate. For larger problems, the ASP may use commercial optimizers, e.g., LINDO \cite{lin2009global}, CPLEX \cite{cplex2009v12}, to obtain optimal solutions. For the coherence, we here skip the details for solving the subproblems and denote the optimal Oracle solution for each subproblem in time slot $t$ as $\f^{*,t}$. The collection $\{\f^{*,t}\}_{t=1}^{T}$ is the Oracle solution to ERR problem. Later in Section \ref{sec:ext}, both exact and approximate solutions for optimization problem in \eqref{opt:subproblem} will be discussed using the framework of \emph{Knapsack problem with Conflict Graphs} (KCG). In addition, the impact of error due to approximation on the performance of the proposed algorithm will be analyzed.

However, in practice, the ASP does not have a priori knowledge on the users' service demand, and therefore the ASP has to make rental decisions $\f^t$ based on the service demand estimation $\hat{\bm{\lambda}}^t$ in each time slot. An online decision-making policy designs certain strategies to choose a rental decision $\f^t$ based on the estimation $\hat{\bm{\lambda}}^t$. The performance of designed policy is measured by utility loss, termed \emph{regret}, compared to the utility achieved by Oracle solution. The expected regret of a policy is defined by:
\begin{align}
\mathbb{E}\left[R(T)\right] = \sum\nolimits_{t=1}^T \left( \mathbb{E}\left[U^t(\f^{*,t};\bm{\lambda}^t)\right]- \mathbb{E}\left[U^t(\f^{t};\bm{\lambda}^t)\right] \right)
\end{align}
Here, the expectation is taken with respect to the decisions made by the decisions made by the decision-making policy and the service demand distribution over context.

\subsection{Context-aware Online Edge Resource Rental Algorithm}
Now, we are ready to present our online decision-making algorithm called \emph{Context-aware Online Edge Resource Rental} (COERR). The COERR algorithm is designed in the framework of CC-MAB. In each time slot $t$, ASP operates sequentially as follows: (i) ASP observes the contexts of all $N$ SBSs $\x^t = \{x^t_n\}_{n\in\mathcal{N}}, x^t_n \in \mathcal{X}$. (ii) ASP determines its rental decision $\f^t$ based on the observed context information $\x^t$ in the current time slot and the knowledge (i.e., the connection between SBS context and service demand) learned from the previous time slots. (iii) The rental decision $\f^t$ is applied. If $f^t_n \neq 0$, the users within the coverage of SBS $n$ can offload computation tasks to SBS $n$ for edge processing. (iv) At the end of the time slot, the number of tasks received by rented SBS $n$ (i.e. $f_n > 0$) is observed $\lambda_n^t$, which is then used to update the service demand estimation $\hat{\lambda}_n(x^t_n)$ for the observed context $x^t_n$ of SBS $n$. The users who cannot access the edge service will offload tasks to the cloud server.    

The context of SBSs is from a continuous space and hence there can be infinitely many contexts for an SBS. It would be extremely laborious, if not impossible, to collect historical demand records and learn a service demand distribution for each possible context. To make the context-aware demand estimation tractable, COERR groups similar contexts and learns the demand pattern for a group of contexts instead of learning the service demand pattern for each context $x\in\mathcal{X}$. The rationale behind this strategy is the following intuition: an SBS will have similar service demand when its contexts are similar. This is a natural assumption in practice and is used in many existing MAB algorithms \cite{muller2017context,chen2018contextual} to facilitate the learning of context-aware service demand. To be specific, COERR groups contexts by partitioning the context space into small hypercubes. The context space $\mathcal{X} = [0,1]^D$ is split into $(h_T)^D$ hypercubes give the time horizon $T$, where each hypercube is $D$-dimensional with identical size $\frac{1}{h_T}\times\dots\times\frac{1}{h_T}$. Here, $h_T$ is an important input parameter to be designed to guarantee algorithm performance. These hypercubes are collected in the context partition $\mathcal{P}_T$. Since the edge system is geographically distributed, different SBSs may exhibit distinct service demand patterns for the same context because of the SBS locations (e.g., considering the time factor, an SBS located in a school zone may have higher service demand during daytime and lower service demand during night while an SBS located in a residential area tends to have lower service demand during daytime and higher service demand at night). Therefore, ASP should learn the service demand for each SBS.  

Now, a key issue is estimating the service demand pattern for context hypercubes at each SBS. Note that COERR runs with \emph{cold-start} and hence it needs to collect the historical service demand data for context hypercubes by renting computation resource at SBSs and observing the received service demand in order to produce accurate demand estimation. Specifically, (i) for each SBS $n\in\mathcal{N}$, ASP keeps counters $C^t_n(p)$, one for each hypercube $p\in\mathcal{P}_T$, up to time slot $t$, indicating the number of times that ASP rents computation resource at SBS $n$ (i.e., $f^\tau_n > 0, \tau < t$) when the context $x^\tau_n$ of SBS $n$ belongs to hypercube $p$, i.e. $x^\tau_n \in p, \tau < t$; (ii) ASP keeps an experience $\mathcal{E}^t_n(p)$ for hypercube $p$ at each SBS $n$ up to time slot $t$ storing the context-demand pair $(x^\tau_n,\lambda^\tau_n)$ when the rental decision $f^\tau_n > 0$ is taken and the context of SBS $n$ satisfies $x^\tau_n \in p$. Fig.\ref{fig:illu_context_space} illustrates an example of context space partition and counter/experience update.
\begin{figure*}[tb]
	\centering
	\includegraphics[width=0.9\linewidth]{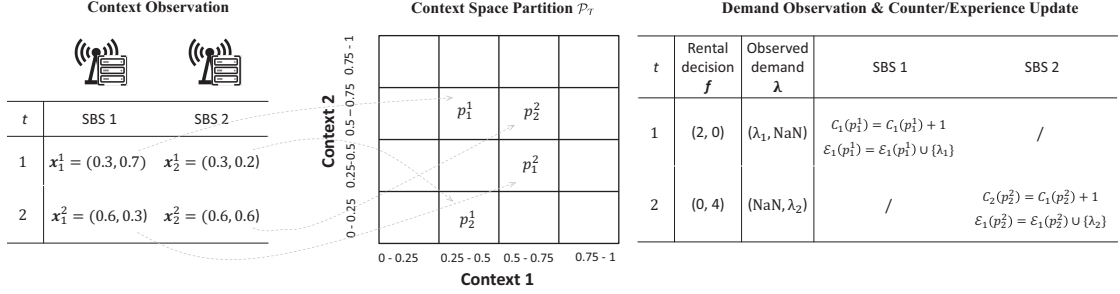}
	\vspace{-0.2 in}
	\caption{Illustration for context space partition and counter/experience update. In the time slot $t=1$, only SBS 1 is rented to host the service and therefore the counter/experience is updated for SBS 1 only.}
	\label{fig:illu_context_space}
		\vspace{-0.25 in}
\end{figure*}

Given the experience $\mathcal{E}^t_n(p)$, the service demand estimation for SBSs $n$ with context $x^t_n$ in hypercube $p$ is obtained by an estimator $\Theta_n$:
\begin{align}
\hat{\lambda}_n(p) = \Theta_n(\mathcal{E}^t_n(p)),
\end{align}
We do not specify the estimator used in COERR since the proposed algorithm is compatible with a variety of estimators. Note that storing all the experience may be unnecessary for certain estimators that can be updated in a recursive manner, e.g., recursive Bayesian estimator \cite{bergman1999recursive} and recursive least square estimator \cite{lee1981recursive}. Usually, a certain amount of historical data is required for an estimator to produce an accurate-enough estimation, which is theoretically characterized by Probably Approximately Correct (PAC) \cite{valiant2013probably} as follows:
\begin{assumption}[PAC Property]\label{ass:PAC}
	For an arbitrary hypercube $p\in\mathcal{P}_T$ at a SBS $n$, the estimator $\Theta_n$ satisfies Probably Approximately Correct (PAC) property below:
	\begin{align}
	\Pr\left\{\lvert \Theta_n(\mathcal{E}^t_n(p))- \mu_n(p) \rvert> \epsilon \right\} \leq \sigma_n(\epsilon,C^t_n(p))
	\end{align}
	where $\mu_n(p) = \mathbb{E}_{x \in p}\left[ \mu_n(x) \right]$ (the expectation is taken on the distribution of context $x$ in hypercube $p$) and $\frac{\partial\sigma(\epsilon,C^t_n(p))}{\partial C^t_n(p)}\leq 0$.
\end{assumption}
The term $\sigma(\epsilon,C^t_n(p))$ is assumed to decrease as $C^t_n(p)$ increases, i.e. $\frac{\partial\sigma(\epsilon,C^t_n(p))}{\partial C^t_n(p)}\leq 0$, which ensures that more historical data will produce a better estimation. The PAC property is critical in guaranteeing the performance of COERR. 

It is worth empathizing that service demand estimation, though important, is not the major challenge to conquer since we can always acquire enough data for each hypercube to produce an accurate estimation if the time horizon $T$ is large. A more challenging issue is to decide in each time slot whether the current demand estimation is good-enough to guide the edge resource rental (referred as \emph{exploitation}) or more service demand data should be collected to improve the demand estimation for a certain hypercube (referred as \emph{exploration}). COERR balances the exploration and exploitation phases during online decision-making in order to maximize the utility of ASP up to a finite time horizon $T$. In addition, COERR also smartly decides the amount of computation resources to rent at different phases to achieve different purposes: in the exploration, COERR utilizes the budget to collect as much service demand data as possible to improve the estimation while in the exploitation, COERR aims to maximize the ASP utility under the budget constraint.     

Algorithm \ref{alg:COERR} presents the pseudo-code of COERR. In each time slot $t$, ASP first observes the context $\x^t = \{x_n^t\}_{n\in\mathcal{N}}$ of all SBSs in $\mathcal{N}$ and determines for each SBS $n$ the hypercube $p^t_n \in \mathcal{P}_T$ to which $x^t_n$ belongs to, i.e. $x^t_n\in p^t_n$ holds. The hypercubes of all SBSs are collected in $\p^t = \{p^t_n\}_{n\in\mathcal{N}}$. The estimated service demand for SBS $n$ in time slot $t$ is obtained by $\hat{\lambda}_n(p^t_n) = \Theta_n(\mathcal{E}^t_n(p^t_n))$. Estimations of all SBSs are collected in $\hat{\bm{\lambda}}^t = \{\hat{\lambda}_1(p^t_1),\hat{\lambda}_2(p^t_2),\dots,\hat{\lambda}_N(p^t_N)\}$. COERR is in either an exploration phase or an exploitation phase. To determine the phase for current time slot, the algorithm checks whether current contexts of SBSs have been sufficiently explored. To this end, we define the set of \emph{under-explored} SBSs $\mathcal{U}^t$ based on the contexts $\x^t$ observed and counters $C^t_n(p^t_n)$ in time slot $t$:
\begin{align}\label{eq:under_explored_set}
	\mathcal{U}^t(\x^t) = \left\{ n\in\mathcal{N} \mid C^t_n(p^t_n) < K(t), x^t_n \in p^t_n\right\}
\end{align}
where $K(t)$ is a deterministic, monotonically increasing control function, which is an input of COERR to determine whether the amount of collected historical data in hypercube $p^t_n$ is large enough to produce an accurate service demand estimation for exploitation in time slot. $K(t)$ has to be designed appropriately based on the estimator property $\sigma(\epsilon,C^t_n(p))$ and the parameter $h_T$ to balance the trade-off between exploration and exploitation (discussed later in Section \ref{sec:parameter_design}).  
\begin{algorithm}[tb]
	\caption{Context-aware Online Edge Resource Rental (COERR)} \label{alg:COERR}
	\begin{algorithmic}[1]
		\State \textbf{Input}: $T$, $h_{T}$, $K(t)$.
		\State \textbf{Initialization}: create partition $\mathcal{P}_{T}$ on context space $\mathcal{X}$; set $C_{n}(p)=0, \mathcal{E}_{n}(p) = \emptyset, \forall n, \forall p \in \mathcal{P}_{T}$; choose an estimator for each SBS $\Theta_n$.
		\For {$t=1,\dots,T$}
		\State The ASP observes the context of each SBS $n\in\mathcal{N}$ and collects them in $\x^t = \{x^t_n\}_{n\in\mathcal{N}}$;
		\State Determine $\p^t=\{p^t_n\}_{n\in\mathcal{N}}, p^t_n\in\mathcal{P}_T$ such that $x^t_{n}\in p^t_{n}$;
		\State Identify under-explored SBSs $\mathcal{U}^t$ as in \eqref{eq:under_explored_set};
		\If {$\mathcal{U}^t \neq \emptyset$}: \Comment{\textit{Exploration}}
		\If{$\sum_{n\in\mathcal{U}^t} w_n (f^{\min}_n) \geq B$}
		\State Select a subset $\mathcal{S}^t \subseteq \mathcal{U}^t$ as in \eqref{eq:explore_pick};
		\State The rental decision at SBS $n$ is $f^t_n = f^{\min}_n \cdot \bm{1}\{n\in\mathcal{S}^t\}$;
		\Else~$\left(\text{i.e.,}\sum_{n\in\mathcal{U}^t} w_n (f^{\min}_n) < B\right)$: Get $\f^t$ by soling the optimization problem in \eqref{opt:semi_explore};
		\EndIf
		\Else~$\left(\mathcal{U}^t = \emptyset\right)$: \Comment{\textit{Exploitation}}
		\State Get the rental decision $\f^t$ by solving the optimization problem in \eqref{opt:exploit};
		\EndIf
		\State Observe service demand receive at SBSs with $f^t_n > 0$;
		\For {each SBS $n$ with $f^t_n > 0$} \Comment{\textit{Update counters, experiences, and estimations}}
		\State Update counters: $C_n(p^t_{n})=C_n(p^t_{n})+1$; \label{line:counter_update}
		\State Update experiences: $\mathcal{E}_n(p^t_{n}) = \mathcal{E}_n(p^t_{n}) \cup (x^t_n,\lambda^t_n)$; \label{line:experience_update}
		\State Update estimations: $\hat{\lambda}_n(p^t_n) = \Theta_n(\mathcal{E}_n(p^t_{n}))$;
		\EndFor
		\EndFor
		\State \textbf{Return}: $\f^t, t = 1,2,\dots,T$.
	\end{algorithmic}
\end{algorithm}

\subsubsection{Exploration}
If the under-explored set is non-empty, i.e., $\mathcal{U}^t \neq \emptyset$, COERR enters the exploration phase. We may have two cases in exploration: (i) If $\sum_{n\in\mathcal{U}^t} w_n(f^{\min}_n) \geq B$, COERR can explore only a subset of SBSs in $\mathcal{U}^t$. Intuitively, we want to collect service demand data for more under-explored SBSs. Therefore, COERR rents only $f^{\min}_n$ at SBSs such that the edge service can be deployed at more under-explored SBSs. Specifically, COERR selects under-explored SBSs sequentially as follows:
\begin{align}\label{eq:explore_pick}
	s_k = \argmin\nolimits_{n\in\mathcal{U}^t\backslash\{s_i\}_{i=1}^{k-1}} w_n(f^{\min}_n)
\end{align} 
If the SBS defined in \eqref{eq:explore_pick} is not unique, ties are broken arbitrarily. The selection ends if the iteration $k$ satisfies $\sum_{i=1}^k w_{s_i}(f^{\min}_{s_i}) \leq B$ and $\sum_{i=1}^{k+1} w_{s_i}(f^{\min}_{s_i}) > B$. The rental decision of ASP at SBS $n$ is $f^t_n = f^{\min}_n \cdot \bm{1}\{n\in\mathcal{S}^t\}$ where $\mathcal{S}^t=\{s_1,\dots,s_k\}$. The selection in \eqref{eq:explore_pick} ensures that the number of under-explored SBSs with $f^t_n > 0$ is maximized. (ii) If $\sum_{n\in\mathcal{U}^t} w_n (f^{\min}_n) < B$, COERR rents computation resource $f^{\min}_n$ at all under-explored SBSs in $\mathcal{U}^t$. Note the there is still $B-\sum_{n\in\mathcal{U}^t} w_n (f^{\min}_n)$ budget left. The rest budget is used to rent computation resources at explored SBSs $n\in \mathcal{N}\backslash \mathcal{U}^t$ based on the current estimation $\hat{\bm{\lambda}}^t$. The rental decision of ASP $\f^t$ in this case can be obtained by:
\begin{subequations}\label{opt:semi_explore}
	\begin{align}
		&\max_{\f^t} ~ U^t(\f^t; \hat{\bm{\lambda}}^t) \\
		\text{s.t.} ~~~ & f^t_n = f^{\min}_n, \forall n \in \mathcal{U}^t \label{cstr:under_explored}\\
		& \eqref{cstr:sub_c1},\eqref{cstr:sub_c2}, \eqref{cstr:sub_c3}
	\end{align}
\end{subequations}  
Constraint \eqref{cstr:under_explored} ensures that the computation resource $f^{\min}_n$ is rented at under-explored SBSs.

\subsubsection{Exploitation}
If the set of under-explored SBSs is empty, i.e., $\mathcal{U}^t = \emptyset$, then COERR enters the exploitation phase in which an optimal rental decision $\f^t$ is determined based on the current service demand $\hat{\bm{\lambda}}^t$. The rental decision $\f^t$ is obtained by solving:
\begin{align}\label{opt:exploit}
	\max\nolimits_{\f^t} ~ U^t(\f^t; \hat{\bm{\lambda}}^t) \quad \text{s.t.} ~ \eqref{cstr:sub_c1},\eqref{cstr:sub_c2}, \eqref{cstr:sub_c3}
\end{align}

\subsection{Performance Analysis}
Next, we give an upper performance bound of COERR in terms of the \emph{regret}. The regret upper bound is derived based on the natural assumption that the service demands received by an SBS are similar when its contexts are similar. This assumption is formalized by the following H\"{o}lder condition \cite{muller2017context,chen2018contextual} for each SBS $n\in\mathcal{N}$.
\begin{assumption}[H\"{o}lder Condition] \label{assu:holder}
	For an arbitrary SBS $n\in\mathcal{N}$, there exists $L>0$, $\alpha>0$ such that for any $x,x^\prime\in\mathcal{X}$, it holds that $|\mu_n(x)-\mu_n(x^\prime)| \leq L \|x-x^\prime\|^{\alpha}$, where $\|\cdot\|$ denotes the Euclidean norm in $\mathbb{R}^{D}$.
\end{assumption}
Note that this assumption is needed for the analysis of regret but the proposed algorithm can still be applied if it does not hold true. In that case, however, a regret bound might not be guaranteed. We aim to design the input parameters $h_T$, $K(t)$ in the proposed algorithm to achieve a sublinear $R(T) = O(T^{\gamma})$ with $\gamma<1$. A sublinear regret bound guarantees that the proposed algorithm has an asymptotically optimal performance since $\lim_{T\to\infty} \frac{R(T)}{T} = 0$ holds. This means that the online decision made by COERR converges to the Oracle solution. 

Since any time slot is either in exploration or exploitation, we divide the regret two parts $R(T) = R_\text{explore}(T) + R_\text{exploit}(T)$, where $R_\text{explore}(T)$, $R_\text{exploit}(T)$ are the regrets due to exploration and exploitation, respectively. These two parts will be bounded separately to get the total regret bound. We first give an upper bound for exploration regret. 

\begin{lemma}\label{lemma:R_explore}
	(Bound of $\mathbb{E}[R_\text{explore}(T)]$.) Given the input parameters $h_T$ and $K(t)$, the regret $\mathbb{E}[R_\text{explore}(T)]$ is bounded by:
	\begin{align*}
		\mathbb{E}[R_\text{explore}(T)] \leq \frac{NB\lambda^{\max}d^{\max}}{w^{\min}}(h_T)^D\lceil K(T)\rceil
	\end{align*} 
	where $\lambda^{\max} = \max_{f^t_n} \lambda^{\max}(f^t_n)$ and $w^{\min} = \min_{f\in F_n,\forall n} w_n(f)$.
\end{lemma}

\begin{proof}
	Suppose time slot $t$ is an exploration phase, then according to the algorithm design, the set of under-explored SBSs is non-empty. Therefore, there must exist $n\in\mathcal{N}$ and a hypercube $p^t_n$ satisfies $C^t_n(p^t_n)<K(t)$. Clearly, there can be at most $\lceil K(t) \rceil$ exploration phases in which computation resources at SBS $n$ are rented by the ASP when its context satisfies $x^\tau_n \in p^t_n, \tau < t$. 
	
	In each of these exploration phase, let $\Psi_n^{\max,t} \triangleq \max_{f_n,f^{\prime}_n \in F_n}|\Delta^t_n(f_n)-\Delta^t_n(f^\prime_n)|$ be the maximum utility loss for one task due to a wrong rental decision $f_n$ at SBS $n$. Recall that the per-task delay reduction is bounded by $\Delta^t_n(f^t_n) \leq d^{\max}, \forall n, \forall t$ and therefore it holds that $\Psi_n^{\max,t}\leq d^{\max}$. Let $\lambda^{\max} = \max_{f^t_n} \lambda^{\max}(f^t_n)$, then the service demand $\lambda_n(x^t_n)$ received by SBS $n$ must be bounded by $\lambda^{\max}$, the maximum utility loss at a SBS is bounded by $\lambda^{\max}d^{\max}$. Let $w^{\min} = \min_{f\in F_n,\forall n} w_n(f)$, the maximum number of SBSs with the rental decision $f^t_n > 0$ is bounded by $B/w^{\min}$. Therefore, the regret incurred in one time slot is bounded by $\lambda^{\max}d^{\max}B/w^{\min}$. Since there are at most $N(h_T)^D\lceil K(T)\rceil$ exploration phases in $T$, the regret incurred by the exploration is bounded by:
	\begin{align*}
		\mathbb{E}[R_\text{explore}(T)] \leq \frac{NB\lambda^{\max}d^{\max}}{w^{\min}}(h_T)^D\lceil K(T)\rceil
	\end{align*} 
	The proof is completed.
\end{proof}
Lemma \ref{lemma:R_explore} shows that the order of $R_\text{explore}(T)$ is determined by the number of hypercubes $(h_T)^D$ in partition $\mathcal{P}_T$ and the control function $K(T)$. 

\begin{lemma}\label{lemma:R_exploit}
	(Bound of $\mathbb{E}[R_\text{exploit}(T)]$.) Given the input parameter $h_T$ and $K(t)$, if the H\"{o}lder condition holds true and the additional condition $2H(t) + 2d^{\max}NLD^{\frac{\alpha}{2}}h_{T}^{\alpha} \leq At^\theta$ is satisfied with some $H(t)>0, A>0, \theta<0$ for all $t$, then $\mathbb{E}[R_\text{exploit}(T)]$ is bounded by:
	\begin{align}\label{eq:R_exploit}
		\mathbb{E}[R_\text{exploit}(T)] \leq \frac{2|\mathcal{F}|B\lambda^{\max}d^{\max}}{w^{\min}}\sum_{t=1}^T\sum_{n\in\mathcal{N}}\sigma_n\left(\frac{H(t)}{d^{\max}N},K(t)\right) + 3d^{\max}NLD^{\frac{\alpha}{2}}h_{T}^{-\alpha}T + AT^{\theta+1}
	\end{align} 
\end{lemma}
\begin{proof}
See in online Appendix \ref{proof:lemma:exploit} \cite{onlineappendix}.
\end{proof}
Lemma \ref{lemma:R_exploit} indicates that, besides the input parameters $h_T$ and $K(t)$, the regret incurred in exploitation also depends on the estimator's PAC property $\sigma_n(\cdot,\cdot)$. Based on the above two Lemmas, we will have the following Theorem for the upper bound of $\mathbb{E}[R(T)]$.
\begin{theorem}\label{theo:regret_bound_raw}
	Given the input parameter $h_T$ and $K(t)$, if the H\"{o}lder condition holds true and the additional condition $2H(t) + 2d^{\max}NLD^{\frac{\alpha}{2}}h_{T}^{\alpha} \leq At^\theta$ is satisfied with some $H(t)>0, A>0, \theta<0$ for all $t$, then $\mathbb{E}[R(T)]$ is bounded by:
	\begin{align*}
		\mathbb{E}[R(T)] \leq & \frac{NB\lambda^{\max}d^{\max}}{w^{\min}}(h_T)^D\lceil K(T)\rceil\\
		& +\frac{2|\mathcal{F}|B\lambda^{\max}d^{\max}}{w^{\min}}\sum_{t=1}^T\sum_{n\in\mathcal{N}}\sigma_n\left(\frac{H(t)}{d^{\max}N},K(t)\right) + 3d^{\max}NLD^{\frac{\alpha}{2}}h_{T}^{-\alpha}T + AT^{\theta+1}
	\end{align*}
\end{theorem}

The regret upper bound in Theorem \ref{theo:regret_bound_raw} is given with any input parameters $h_t$, $K(t)$ and applied estimators. In addition, there is an additional condition $2H(t) + 2d^{\max}NLD^{\frac{\alpha}{2}}h_{T}^{\alpha} \leq At^\theta$ should be satisfied when designing algorithm parameters $h_T$. However, we cannot give a specific design of $h_T$ here to guarantee the sublinear regret since it depends on the PAC property of the applied estimator. In the next subsection, we will design input $h_T$ and $K(t)$ based on the PAC property $\sigma_n(\cdot,\cdot)$ of a Maximum Likelihood Estimator, which satisfy the additional condition posed in Theorem \ref{theo:regret_bound_raw} and guarantee a sublinear regret $O(T^\gamma),\gamma<1$. Other parameters $H(t)$, $A$, $\theta$ are not determinative which will be later shown in parameter design.

\subsection{Example: Maximum Likelihood Estimator}\label{sec:parameter_design}
Note that the regret depends partially on the estimator property $\sigma_n(\cdot,\cdot)$ and hence we need to specify the estimators used by SBSs before designing the algorithm parameters $h_T$ and $K(t)$. Here, we take Maximum Likelihood Estimation (MLE) as an example. The purpose of a MLE estimator $\Theta_n(\mathcal{E}(p^t_n))$ is to estimate the expected service demand $\mu_n(p^t_n)$ for hypercube $p^t_n$. We assume that the historical service demands $\lambda_n(x^\tau), x^\tau\in p^t_n$ collected in $\mathcal{E}(p^t_n)$ follow a normal distribution denoted by $\mathcal{N}\left(\mu_n(p^t_n),\delta_n^2(p^t_n)\right)$, where $\delta_n^2(p^t_n)$ is the standard deviation. Then, an unbiased estimation for $\mu_n(p^t_n)$ using MLE is:
\begin{align}\label{eq:MLE}
	\hat{\lambda}_n(p^t_n) = \frac{1}{C^t_n(p^t_n)}\sum\nolimits_{(x,\lambda)\in \mathcal{E}(p^t_n)}\lambda
\end{align}
Note that the normal distribution of historical service demand in $\mathcal{E}(p^t_n)$ is only used for deriving the above MLE estimator. COERR can be applied other historical data distributions, but the unbiased MLE estimator can be different accordingly. The MLE estimator in \eqref{eq:MLE} guarantees the following PAC condition based on the Chernoff-Hoeffding bound \cite{hoeffding1963probability}:
\begin{align}
	\Pr\left(\hat{\lambda}^t_n(p^t_n)- \mu_n(p^t_n)>\epsilon\right)\leq \sigma_n(\epsilon,C^t_n(p^t_n)) = e^{-\frac{2C^t_n(p^t_n)\epsilon^2}{(\lambda^{\max})^2}}
\end{align}
and it holds that $\frac{\partial\sigma(\epsilon,C^t_n(p))}{\partial C^t_n(p)}\leq 0$. Now, we can design $h_T$ and $K(T)$ to ensure a sublinear regret of COERR.
\begin{theorem}[Regret upper bound] \label{theo:regret_bound}
	Let $h_T = \lceil T^{\frac{1}{3\alpha+D}}\rceil$ and $K(t)= t^{\frac{2\alpha}{3\alpha+D}}\log(t)$. If the proposed algorithm runs with these parameters, SBSs use MLE for estimation, and the H\"{o}lder condition holds true, then the leading order of the regret $\mathbb{E}[R(T)]$ is:
	\begin{align*}
		O\left(\dfrac{2^DNB\lambda^{\max}d^{\max}}{w^{\min}} T^{\frac{2\alpha+D}{3\alpha+D}}\log(T)\right).
	\end{align*}
\end{theorem}
\begin{proof}
	See in online Appendix \ref{proof:theo:regret_bound} \cite{onlineappendix}.
\end{proof}
The leading order of regret upper bound given in Theorem \ref{theo:regret_bound} is sublinear. In addition, the regret bound is valid for any $T$ and therefore providing a bound on the performance loss for any time horizon. This also can be used to characterize the convergence speed of COERR. However, we see that the order of upper bound regret can be close to 1 when the dimension of context space $D$ is large. In this case, the learner may need to apply dimension reduction techniques based on empirical experience to cut down the context dimension.

Though the algorithm parameter $h_T$ and the regret upper bound is given based on a known time horizon $T$, COERR can be easily extended to work with unknown time horizon with the assistance of doubling-trick \cite{cesa1997use, tekin2017adaptive}. The key idea of doubling-trick is to partition the time into multiple phases $(j = 1,2,3,...)$ with doubling length ($T_1, T_2, \cdots$), e.g., if the length of phase is $T_1 = T$, then the length of $j$-th phase is $2^{j-1}T$. In each phase, COERR is run from scratch without using any information from the previous phase. A salient property of doubling-trick is that it does not change the order of the upper regret bound.   

\subsection{Complexity and Scalability}
The memory requirement of COERR is mainly determined by the number of counters $C^t(p)$ and experiences $\mathcal{E}^t(p)$ maintained for hypercubes. Since the counter is an integer for each hypercube, its memory requirement is determined by the number of created hypercubes. The experience $\mathcal{E}^t(p)$ is a set of observed service demand records up to time slot $t$ which needs a higher memory requirement. However, storing all historical data is actually unnecessary since most estimators, including MLE in \eqref{eq:MLE}, can be updated in a recursive manner. Therefore, the ASP only needs to keep current service demand estimation for a hypercube which is a floating point number. If COERR is run with the parameters in Theorem \ref{theo:regret_bound}, the number of hypercubes is $(h_{T})^{D} = \lceil T^{\frac{1}{3\alpha + D}}\rceil^{D}$. Hence, the required memory is sublinear in the time horizon $T$. This means that when $T \to \infty$, COERR would require infinite memory. Fortunately, in the practical implementations, ASP only needs to keep the counters and experiences of hypercubes which at least one of the observed contexts belongs to. Therefore, the number of counters and experiences to keep is actually much smaller than the analytical requirement. 

\section{Extension: Solutions for Subproblems}\label{sec:ext}
\subsection{Exact and Approximate Solutions for Sub-problems}
In this section, we discuss in detail the solutions for optimization problems in \eqref{opt:subproblem}, \eqref{opt:semi_explore}, and \eqref{opt:exploit}. Since these optimization problems have the same form, we take the Oracle subproblem \eqref{opt:subproblem} as an example. Note that the problem is solved for each time slot $t$, the time index is dropped in this section for ease of notation. The subproblem is a combinatorial optimization which can be formulated as a Knapsack problem \cite{martello1990knapsack}. The Knapsack problem is a classic combinatorial optimization: given a set of items, each with a weight and value, determine the items to include in a collection such that the total weight is less than or equal to a given limit and the total value is as large as possible. In ERR subproblems, each rental decision at a SBS is an item in the Knapsack problem: for an ``item'' $f_n$, its ``weight'' is the rental cost $w_n(f_n)$ and its ``value'' is the utility gain $u_n(f_n,\mu_n(x))$ ($x$ is the context of SBS $n$ in a certain time slot), and the limit is ASP budget $B$. However, the standard formulation of Knapsack problem cannot exactly capture the ERR problem since the ASP can only take one rental decision for one SBS, which means items associated to one SBS cannot be included at the same time. Such an extension of standard Knapsack problem with addition conflict restrictions, stating that from a certain set of items at most one item can be selected, is known as the \emph{Knapsack problem with conflict graph} (KCG). In the following, we formulate the subproblem as KCG problem and discuss its solutions.

These conflict constraints is represented by a undirected graph $G = (V,E)$.
\begin{itemize}
	\item $V$ (Vertices): each rental decision at a SBS $f\in F_n, \forall n$ corresponds to a vertex in the undirected graph $G$.
	\item $E$ (Edges): for an arbitrary pair of vertices $f,f^\prime \in V$, add an edge $e(f,f^\prime)$ between $f$ and $f^\prime$ if $f,f^\prime \in F_n$ are rental decisions for a same SBS.
\end{itemize}
The vertices/items in $f_k\in V$ are indexed by $k = 1,2,\dots,K$ and for a vertex $f_k\in F_n$, we define its weight as $b_k = w_n(f_k)$ and its value as $z_k = \max\{\mu_n(x),\lambda^{\max}\} \Delta_n(f_k)$. In addition, we introduce an indicator $y_k\in\{0,1\}$ for each vertex $f_k$ indicating whether item $f_k$ is taken ($y_k =1$) or not ($y_k =0$). Then, the KCG for subproblem can be written as:
\begin{subequations}
	\begin{align}
		\max ~& \sum\nolimits_{k=1}^K z_ky_k\\
		\text{s.t.}~ & \sum\nolimits_{k=1}^K b_ky_k \leq B \\
		& y_k + y_j \leq 1, ~\forall ( y_k, y_j) \in E, ~ k,j \in \{1,2,\dots,K\}\\
		& y_k \in \{0,1\}, ~k =1,2,\dots,K
	\end{align}
\end{subequations}
KCG is a well-investigated problem. Several existing algorithms, e.g., Branch-and-Bound \cite{bettinelli2017branch}, can be directly used to derive an exact solution for KCG problem. If an exact solution for each KCG/subproblem is obtained. Then, COERR can provide the expected performance as analyzed in the previous section. However, these exact algorithms can be computational-expensive when the number of items is large and therefore their runtime may become a bottleneck in certain applications (though the runtime is less likely to be an issue in our ERR problem since the time scale of the considered problem is relatively large, e.g., several hours). To facilitate the solution of KCG, approximation algorithms are studied to efficiently derive approximate solutions in polynomial runtime. Next, we will discuss the performance of the proposed algorithm when approximate solutions are derived for subproblems.

\subsection{Performance Analysis with Approximate Solutions}
We assume that the approximation algorithm guarantees a performance bound ($\delta$-approximation) compared to the optimal solution as define below:   
\begin{definition}[$\delta$-approximation]\label{def:delta_appro}
	An approximation algorithm is a $\delta$-approximation if the objective value $U^\delta$ achieved by the approximate solution $\f^\delta$ satisfies $\delta U^\delta \geq U^*, \delta>1$ where $U^*$ is the optimal object value achieve by a optimal solution $\f^*$.
\end{definition} 
Definition \ref{def:delta_appro} indicates that a $\delta$-approximation algorithm achieves no less than $\frac{1}{\delta}$ of the optimum. Many existing approximation algorithms can be directly applied, e.g., Fully Polynomial Time Approximation Schemes (FPTAS) \cite{pferschy2017approximation}, to solve the KCG problems. The assumption of $\delta$-approximation prevents the approximate solution from being arbitrarily bad and enables the performance analysis for COERR.

Now we are ready to analyze the performance of proposed algorithm with approximate solution. From Theorem \ref{theo:regret_bound}, we see that the leading order of the regret upper bound is mainly determined by the exploration regret $\mathbb{E}[R_\text{explore}(T)]$. A sublinear upper bound of exploration regret is derived by limiting a sublinear number of time slots that COERR enters the exploration phase with properly designed $h_T$ and $K(t)$. Note that COERR is either in exploration or exploitation, a sublinear number of exploration slots indicates that the number of exploitation slots is non-sublinear. In this case, it is difficult, if not impossible, to guarantee a sublinear regret with approximate solutions even if we have perfect estimation in each exploitation: due to the $\delta$-approximate, the worst performance loss of approximate solution with perfect estimation in one time slot is $\frac{\delta-1}{\delta}U^{*,t}$. Let $T_\text{exploit}$ be number of exploitation slots which is non-sublinear, the upper bound of exploitation regret (with approximate solutions) must be larger than $\frac{\delta-1}{\delta}T_\text{exploit}\mathbb{E}[U^{*,t}]$ which is also non-sublinear. To address this problem, we slightly change the definition of regret by defining the $\delta$-regret below:
\begin{align}\label{eq:regret_appro}
R^{\delta}(T) = \sum\nolimits_{t=1}^T \left(\frac{1}{\delta}U^t(\f^{*,t};\bm{\lambda}^t)- U^t(\f^{\delta,t};\bm{\lambda}^t)\right)
\end{align}
The rental decision $\f^{*,t}$ is still the optimal Oracle solution for subproblems in \eqref{opt:subproblem}. The rental decision $\f^{\delta,t}$ is the online decisions made by the proposed algorithm with approximation algorithm, i.e., solutions to the optimization problem in \eqref{opt:semi_explore} during exploration and the optimization problem in \eqref{opt:exploit} is approximated by a $\delta$-approximation algorithm. In \eqref{eq:regret_appro}, the online decisions derived by COERR with $\delta$-approximation algorithm is actually compared by the lower bound of approximated Oracle solution (i.e., Oracle also use a $\delta$-approximation algorithm to solve the subproblem in \eqref{opt:subproblem}. Such a definition of regret is often used in MAB framework where optimal solution cannot be derived in each round \cite{chen2018contextual}. 
\begin{theorem}[$\delta$-regret upper bound] \label{theo:delat_regret_bound}
	If the proposed algorithm is run with parameters and conditions given in Theorem \ref{theo:regret_bound} and a $\delta$-approximation is applied for optimization, then the leading order of $\delta$-regret $\mathbb{E}[R^{\delta}(T)]$ is:
	\begin{align*}
	O\left(\dfrac{2^DNB\lambda^{\max}d^{\max}}{\delta w^{\min}} T^{\frac{2\alpha+D}{3\alpha+D}}\log(T)\right).
	\end{align*}
\end{theorem}
\begin{proof}
	See in online Appendix \ref{proof:theo:delat_regret_bound} \cite{onlineappendix}.
\end{proof}
Theorem \ref{theo:delat_regret_bound} indicates that our algorithm is able to work well even if the subproblem in each time slot can only be approximately solved and a sublinear $\delta$-regret can be achieved based on the performance guarantee of $\delta$-approximation algorithms.
\section{Experiments}\label{sec:experiment}
In this section, we carry out systematic experiments in a real-world dataset to verify the efficacy of the proposed algorithm.  
\subsection{Experiment Setup}
We use the real-word service demand trace collected by the Grid Workloads Archive (GWA) \cite{iosup2008grid}. The GWA datasets record the task requests received by large-scale multi-site infrastructures (girds) that provide computational support for e-Science. The experiment is mainly run on the GWA dataset, AuverGrid, which collects around 400,000 task requests of 5 grids. To fit the AuverGird data in our ERR context, we assume each grid corresponds to an SBS in the edge network. In some parts of the experiments, we combine other GWA datasets with AuverGrid to increase the number of sites and show the impact of SBS numbers on the algorithm performance. Each task request record has a \emph{``SubmitTime''} (in second) that indicate the time of task arrival and a \emph{``RunSiteID''} that indicates the site for task execution. The rental decision cycle is set as 3 hrs. With this information, we are able to analyze the service demand trace at each SBS. Fig.\ref{fig:real_trace} depicts the service demand trace of three SBSs. It can be observed that the demand patterns are different at different SBSs and hence it is necessary to learn the service demand pattern for each SBS. 
\begin{figure}[htb]
	\centering	
	\subfigure[Service demand trace.]{\label{fig:real_trace}
		\includegraphics[width=0.45\linewidth]{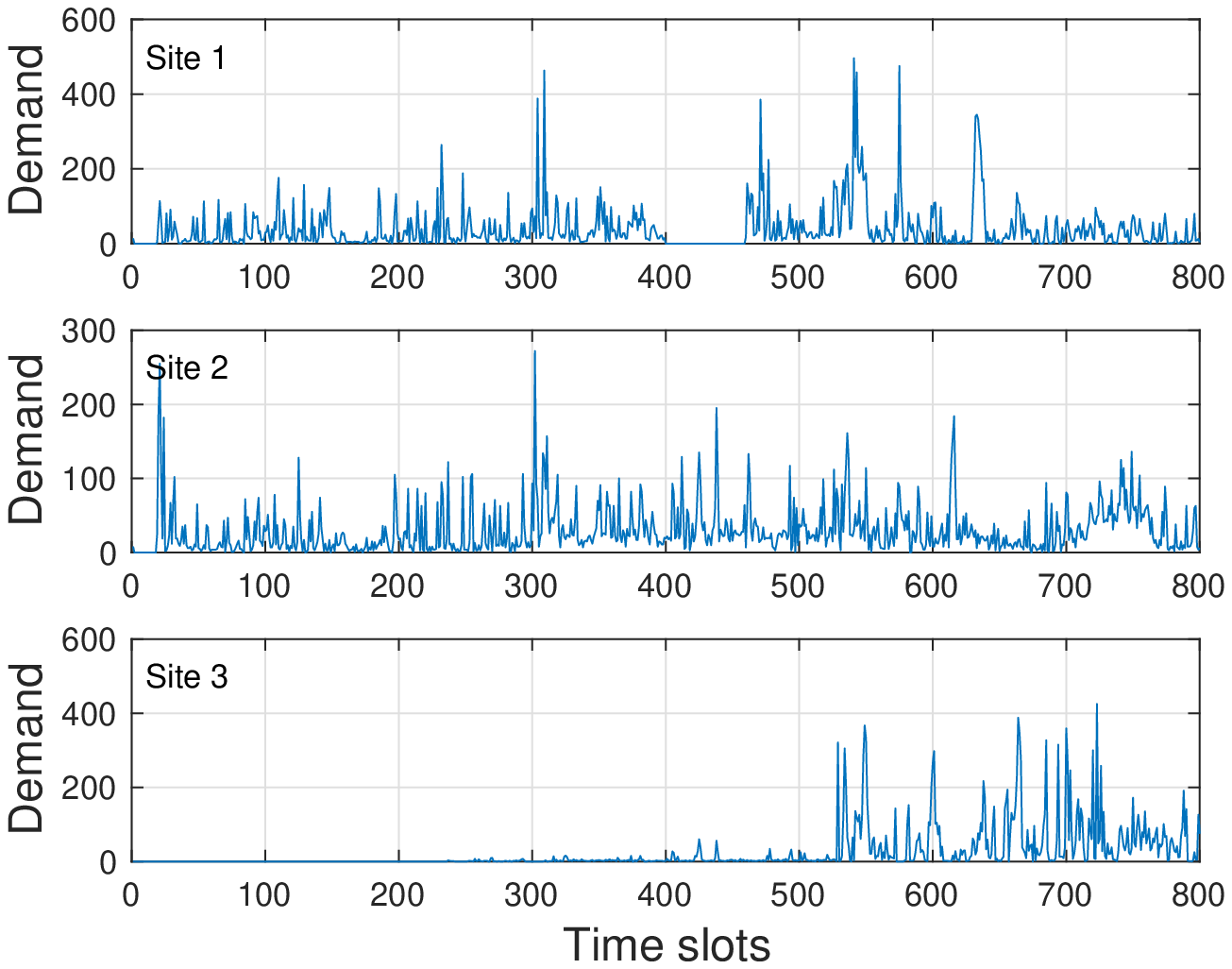}}
	\subfigure[Context-aware service demand.]{\label{fig:context_demand}
		\includegraphics[width=0.5\linewidth]{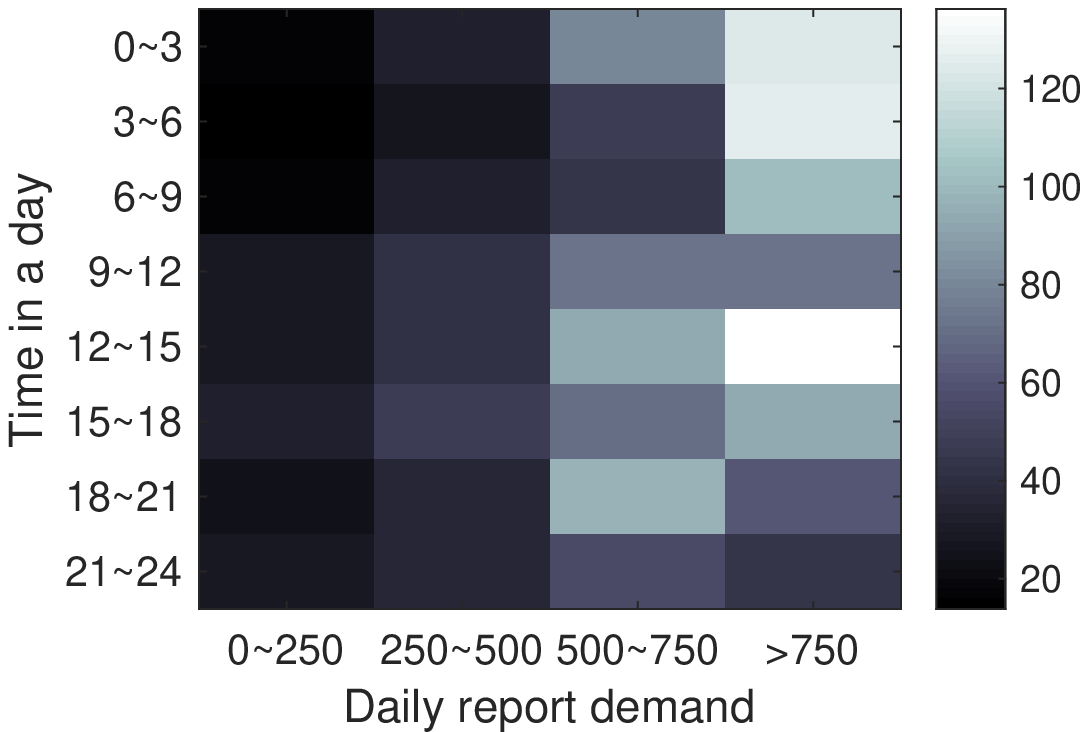}}
	\caption{Real-world service demand. Fig.\ref{fig:real_trace} shows the service demand of three sites in AuverGrid: Site 1: clrlcgce01, Site2: clrlcgce02, Site2: clrlcgce03. Fig.\ref{fig:context_demand} shows the expected service demand for each hypercube maintained by Site 1 (notice that the partition is created as an example  and may be different from the partition designed by the proposed algorithm).}
	\label{fig:real_data}
\end{figure}
The context space of SBSs has two dimension: \emph{``time in a day''} and \emph{``daily report demand''}. The context \emph{``time in a day''} indicates the time when a rental decision is made, and the context \emph{``daily report demand''} is the total service demand received by a SBS in the previous day which is provided by the site daily report. Fig.\ref{fig:context_demand} shows the expected service demand of hypercubes in the context partition of Site 1. We see that the service demand is closely related to the considered contexts. The optimization problems in \eqref{opt:semi_explore}, and \eqref{opt:exploit} are transformed into KCG and solved using Brunch-and-Bound algorithm \cite{bettinelli2017branch}. The computing resource at edge server is discretized as Virtual Machines (VMs) and the rental decision is the number of VMs to rent at SBSs: $\f^t_n \in F_n = \{0,2,4,6\}, \forall n$. The processor capacity of each VM is 2GHz. Therefore, if the rental decision $f^t_n = 2$ then the rented processor capacity at SBS $n$ is 4GHz. Other important parameters are given in Table \ref{tab:exp_parameter}. 
\begin{table}[htb]
	\centering
	\caption{Algorithm complexity}\label{tab:exp_parameter}
	\begin{tabular}{l|l} 
		\hline
		\textbf{Parameter} &  \textbf{Value} \\
		\hline
		Input data size of one task, $s$ & 1MB\\
		Required CPU cycles for one task, $c$ & $10^9$\\
		Pricing mapping function, $w_n(f_n)$ & $w_n(f_n) = 1\cdot f_n$\\
		Maximum service demand processed at SBSs, $\lambda^{\max}(f_n)$ & $\lambda^{\max}(f_n) = 150\cdot f_n$\\
        Path-loss with random shadowing & $P_L = 20\log(d[\text{km}]) + 28 + N(0,5^2)$\\
        Expected wireless transmission rate of SBSs & 5Mbps\\
        Expected wireless transmission rate of MBS & 2Mbps\\
        Bandwidth, $W$ & Spectrum Allocation Scheme \cite{chandrasekhar2009spectrum}\\
        Dimension of context space, $D$ & 2\\
        $\alpha$ in H\"{o}lder & 1 \\
        Time horizon $T$ & 2700\\
        $h_T$ in COERR & 5\\
		\hline
	\end{tabular}\\
\end{table}

The proposed algorithm COERR is compared with following benchmarks:\\
1) Oracle: the Oracle algorithm knows precisely the expected service demand of SBS with any observed context. In each time slot, Oracle chooses rental decisions at SBSs to maximize the ASP utility as in \eqref{opt:subproblem} based on the expected service demand of observed context.\\
2) Combinatorial UCB (CUCB) \cite{chen2013combinatorial}: CUCB is developed based on a classic MAB algorithm, UCB1. The key idea of CUCB create combinations of rental decisions at all SBSs to enumerate all ASP's rental decision $f$. CUCB runs in the UCB1 framework with feasible ASP rental decisions $\f$ that satisfies $\sum_{n \in \mathcal{N}} w_n(f_n) \leq B$ and learns the expected utility for feasible $\f$ overtime.\\
3) LinUCB  \cite{chu2011contextual}: LinUCB considers SBSs' context when running CUCB. LinUCB also learns the expected utility for feasible rental decision $\f$, but LinUCB now observes the context of SBSs and assume the expected utilities of rental decisions linearly depend on the SBSs' context.\\
4) COERR-ORX: COERR-ORX (Zero or X) is a variant of the proposed algorithm COERR. In COERR-ORX, ASP only chooses where to rent computation resource and does not decide how much to rent, i.e., if ASP chooses to rent computation at SBS $n$, it can only take one rental decision  $f_n = X$. Such edge resource rental problem has been considered in \cite{chen2018spatio}\\
5) Random: The algorithm simply chooses one feasible ASP rental decision in each time slot.

\subsection{Results and Discussions}
\begin{figure}[htb]
	\centering	
	\subfigure[Cumulative utility.]{\label{fig:cum_utility}
		\includegraphics[width=0.48\linewidth]{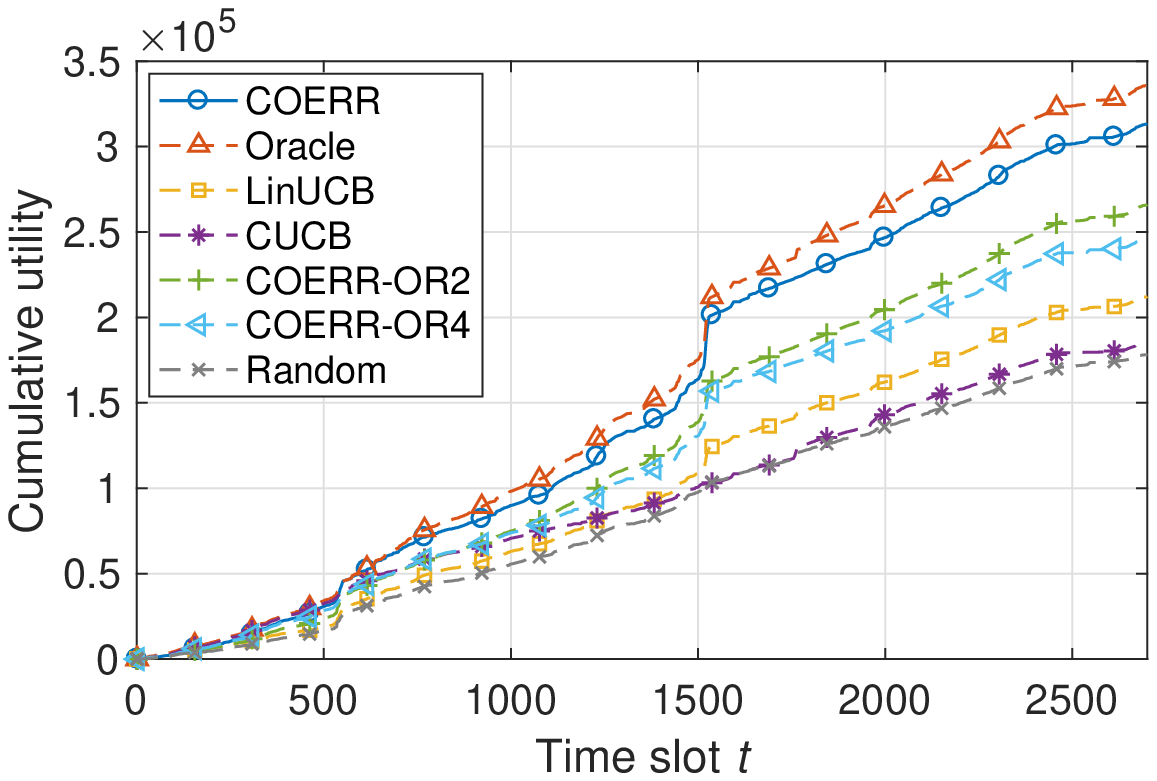}}
	\subfigure[Regret.]{\label{fig:regert}
		\includegraphics[width=0.48\linewidth]{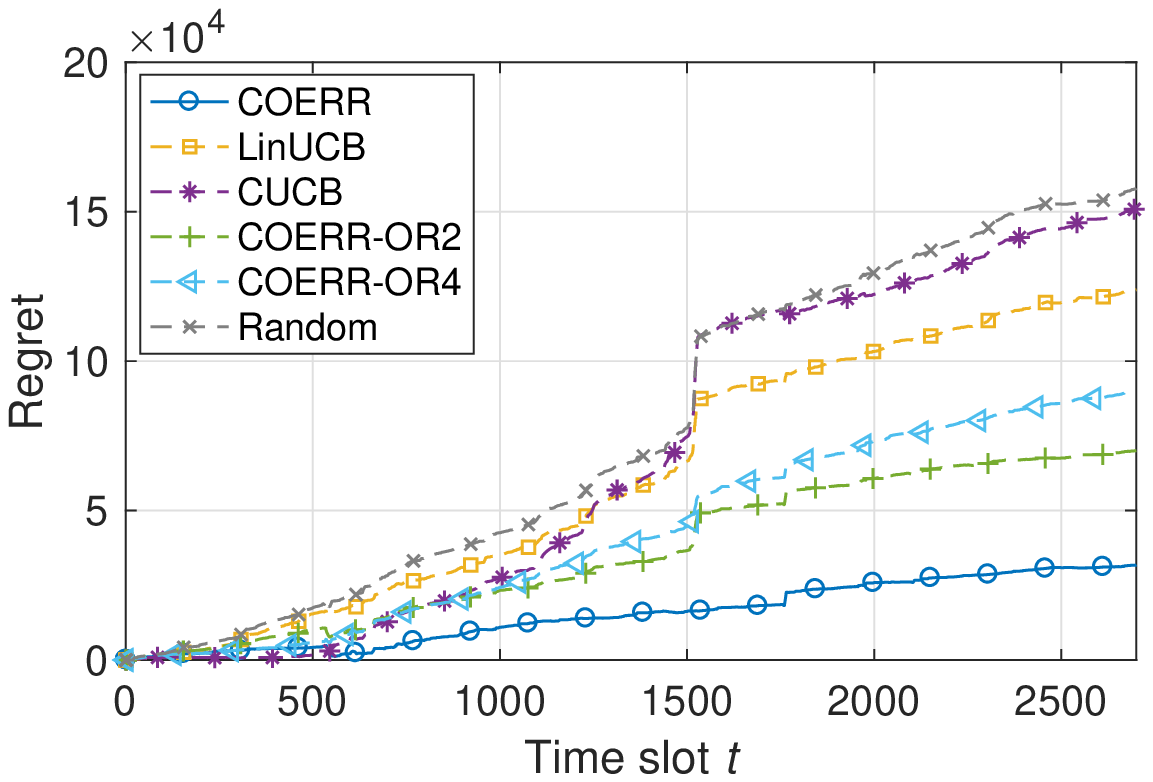}}
	\caption{Comparison on cumulative utilities.}
	\label{fig:runtime}
	\vspace{-0.15 in}
\end{figure}
\subsubsection{Comparison on Cumulative Utilities}
Fig.\ref{fig:runtime} shows the cumulative utilities and rewards achieved by COERR and the other 6 benchmarks during 2,700 time slots. For the cumulative utility in Fig.\ref{fig:cum_utility}, we see that Oracle, as expected, achieves the highest cumulative utility which gives an upper bound to the other algorithms. Among the others, COERR significantly outperforms the other benchmarks and achieves a close-to-Oracle cumulative utility. The benefit of considering the context of SBSs can be appreciated by comparing the performance of context-aware algorithms (COERR, LinUCB and, COERR-ORX) and context-unaware algorithms (CUCB and Random). In addition, we see that the cumulative utility of CUCB is almost the same as the random algorithm. The malfunction of CUCB is due to two reasons:(i) a CUCB arm is a combination of rental decisions at all SBSs and hence the CUCB arm set can be very large. This means CUCB can be easily stuck in the exploration. (ii) CUCB fails to capture the connection between context and service demand. Further analyzing the cumulative utility achieved by LinUCB, we know that considering the context for each possible CUCB arm is not effective to produce a good result due to the large arm set. Comparing the performances of COERR-OR2, COERR-OR4, and COERR, we see that offering more rental decision options at SBSs helps the ASP efficiently utilize its budget and results in a higher cumulative utility. 

Fig.\ref{fig:regert} explicitly depicts the regret incurred by the 6 algorithms. It clearly shows that the proposed algorithm incurs only a sublinear regret (the discontinuity point around slot 1750 is due to the service demand burst at certain sites). 

\subsubsection{Impact of Budget}
Fig. \ref{fig:budget} shows the cumulative utilities achieved by Oracle, COERR, LinUCB, and Random in a total of 2700 time slots. It can be observed that COERR achieves higher cumulative utility compared to LinUCB and Random. In addition, the cumulative utilities achieved by all four algorithms grow with the increase in ASP budget. The reason is intuitive: a larger budget allows the ASP to rent more resource at more SBSs, which means more users can access the edge service and enjoy the low service delay. It is worth noticing that the regrets incurred by COERR, LinUCB, and Random decrease as the budget increase. This is because the ASP can simply place application service at all SBSs without judicious decisions. Though the budget distribution among the SBS may not be optimal, it can avoid large utility loss by using the cloud server.

\begin{figure}
	\centering
	\begin{minipage}{0.48\textwidth}
		\centering
		\includegraphics[width=0.93\linewidth]{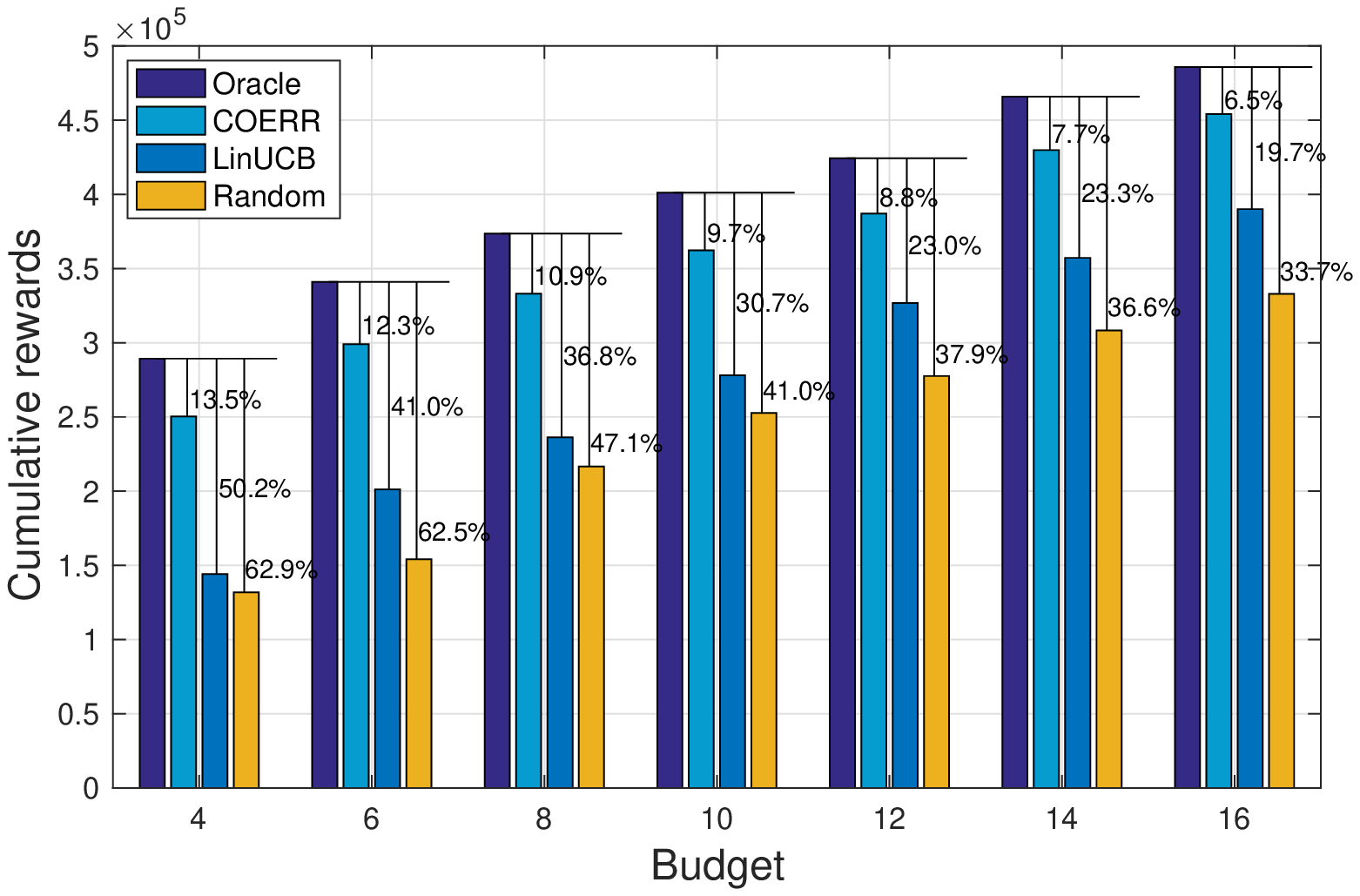}
		\vspace{-0.2 in}
		\caption{Impact of ASP budget.}
		\label{fig:budget}
	\end{minipage}
	\begin{minipage}{0.48\textwidth}
		\centering
			\includegraphics[width=1\linewidth]{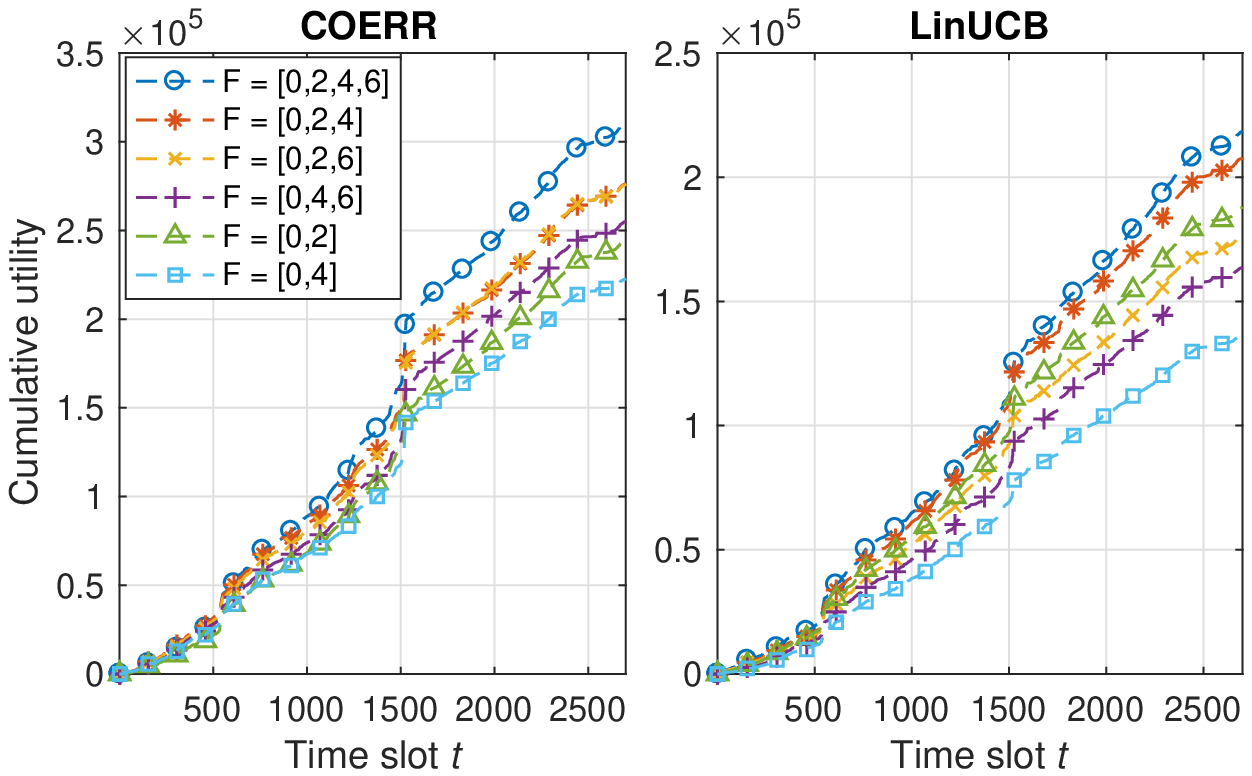}
			\vspace{-0.4 in}
			\caption{Impact of rental decision set.}
			\label{fig:vary_F}
	\end{minipage}
\vspace{-0.2 in}
\end{figure}

\subsubsection{Impact of Rental Decision Set}
Fig.\ref{fig:vary_F} shows the cumulative utility of COERR and LinUCB under different rent decision sets $F$. By comparing the cumulative utilities under rental decision sets with the same size, e.g., $F_n = [0,2,4]$, $F_n = [0,2,6]$, and $F_n = [0,4,6]$, we see that both COERR and LinUCB can achieve higher cumulative utility with smaller $f^{\min}_n$. This is because the ASP can allocation its budget more flexible among SBSs with smaller $f^{\min}_n$. Also, COERR can explore more under-explored arms in one exploration phase with a smaller $f^{\min}_n$, which improves the efficiency of exploration and reduce the regret. 

In addition, we see that COERR achieves a higher cumulative utility when the size of rental decision set $|F|$ is larger. A larger rental decision set $F$ loosens the constraint in per-slot problem \eqref{opt:subproblem}, e.g., the constraint in \eqref{cstr:sub_c2} becomes looser if we change the rental decision set $F_n = [0,2,4]$ to $F_n = [0,2,4,6]$. Therefore, we may have a higher utility in each exploitation phase with $F_n = [0,2,4,6]$. By contrast, a larger $F$ is not always better for LinUCB, e.g., the cumulative utility of $F_n = [0,2]$ is larger than that of $F_n = [0,2,6]$. This is because LinUCB creates more arm with larger $F$, which tends to incur higher regret.

\subsubsection{Running More SBSs} We also vary the number of SBSs in the considered edge system. Since the AuverGrid dataset only records the task request received by 5 distributed sites, we merge it with another GWA dataset, SHARCNET, to get real-world service demand traces for more sites. The merged dataset is used to generate service demand traces for a total of 10 SBSs. The performances of COERR and other benchmarks on the merged dataset are shown in Fig.\ref{fig:runtime_10} where Fig.\ref{fig:cum_utility_BS10} depicts the cumulative utility during runtime and Fig.\ref{fig:regert_BS10} depicts regret. The general trend of cumulative utility in Fig.\ref{fig:cum_utility_BS10} is similar to that in Fig.\ref{fig:cum_utility} and it is can be clearly observed that COERR achieves a sublinear regret. 
\begin{figure}[htb]
	\centering	
	\subfigure[Cumulative utility.]{\label{fig:cum_utility_BS10}
		\includegraphics[width=0.48\linewidth]{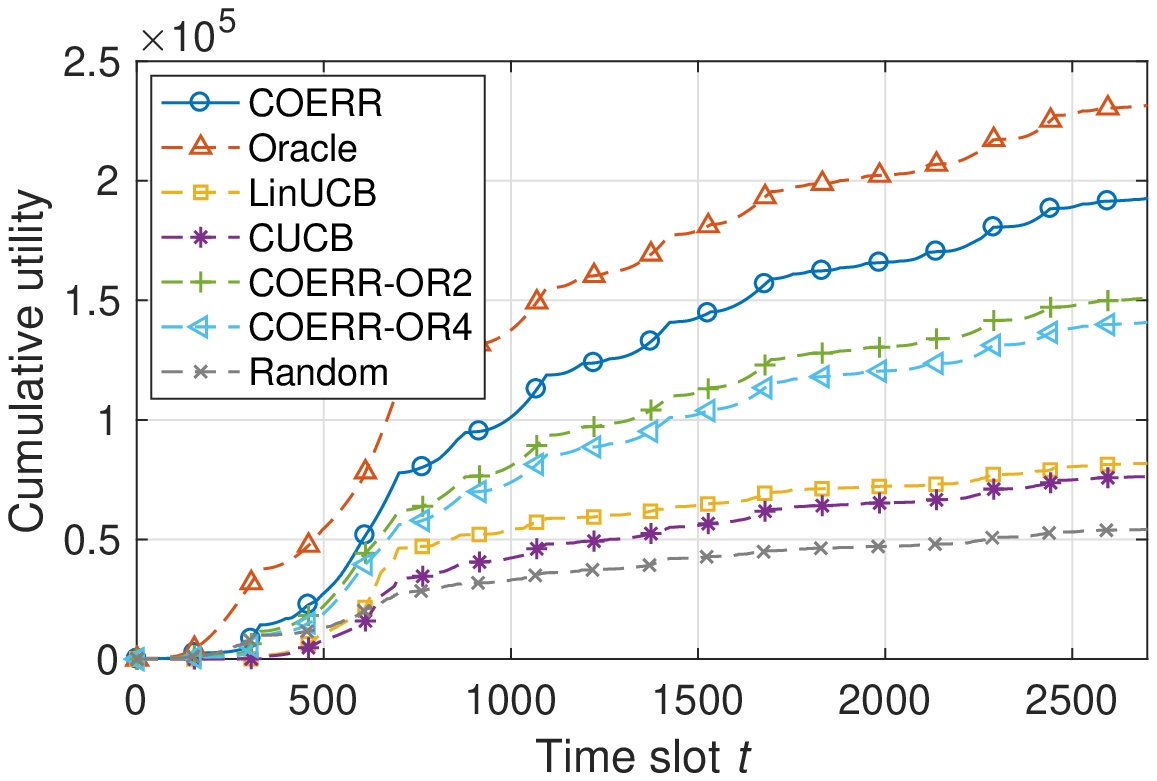}}
	\subfigure[Regret.]{\label{fig:regert_BS10}
		\includegraphics[width=0.48\linewidth]{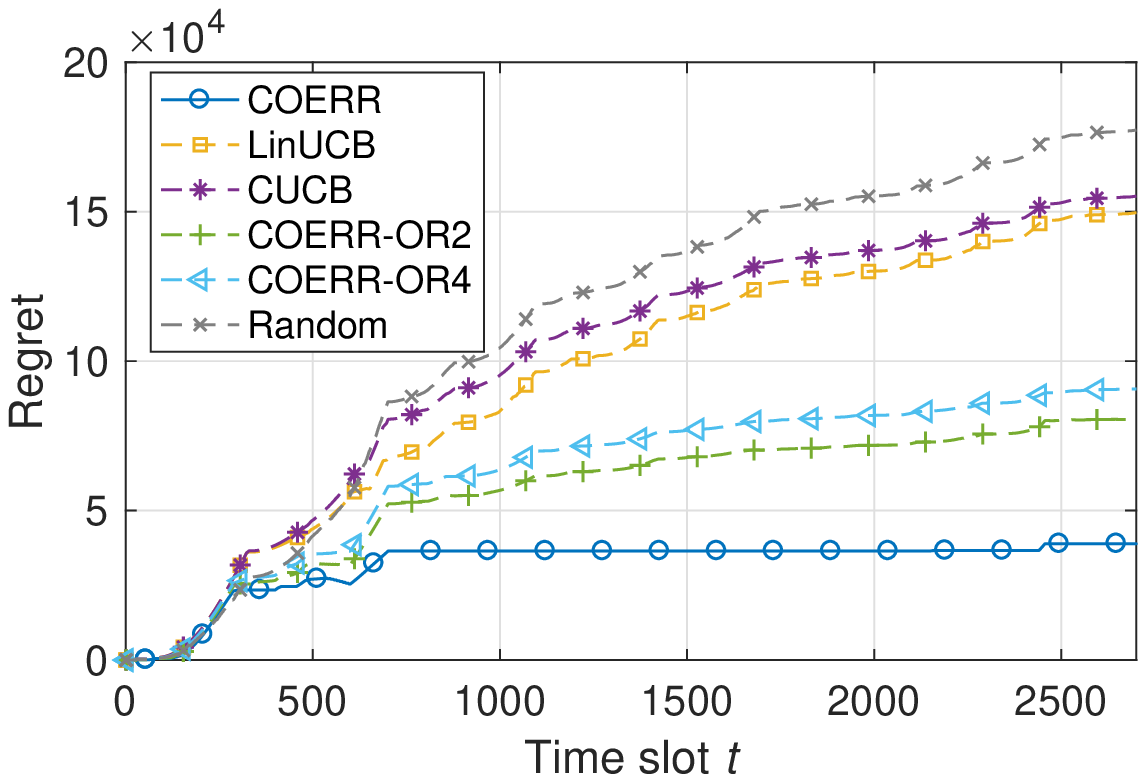}}
	\caption{Runtime performance with 10 SBSs.}
	\label{fig:runtime_10}
\end{figure}

Comparing Fig.\ref{fig:cum_utility_BS10} to Fig.\ref{fig:cum_utility}, we see that COERR incurs a larger regret when running with 10 SBSs. To further analyze the impact of SBS number on the regret, we show the cumulative reward achieved by Oracle, COERR, and CUCB in 2,700 time slots in Fig.\ref{fig:vary_N} and calculate their regrets. There is a general trend that both COERR and CUCB incur a larger regret when there are more SBSs in the edge system. This is because the number of hypercubes created by COERR  and the number of ASP rental decisions created by CUCB become larger when there are more SBSs, which means COERR and CUCB need to spend more time slots in exploration and hence tend to incur larger regret. In addition, we see that the regret of COERR grows slower with the increase in SBS number compared to that of CUCB. The reason for this is that the number of hypercubes for COERR to explore is a linear function of $N$ whereas the number of ASP rental decisions for CUCB is an exponential function of $N$. Table \ref{tab:alg_complexity} shows the number of hypercubes and the number of ASP rental decisions for three experiment setting. Therefore, COERR has better scalability for edge system compared to CUCB. 

\begin{figure}[htb]
		\centering
		\includegraphics[width=0.45\linewidth]{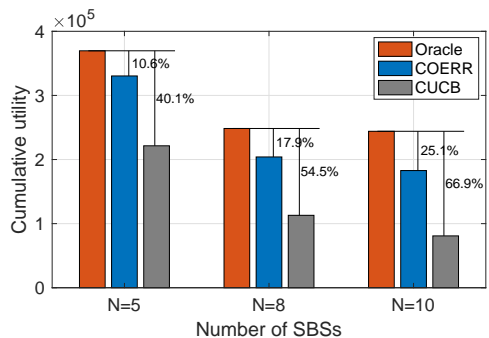}
		\caption{Impact of SBS number on cumulative utility.}
		\label{fig:vary_N}
\end{figure}

\begin{table}[htb]
	\centering
	\caption{Algorithm complexity}\label{tab:alg_complexity}
	\begin{tabular}{l|ccc} 
		\hline
		\textbf{\# SBSs ($N$)} & $N=5$ & $N=8$ & $N=10$ \\
		\hline
		\textbf{\# hypercubes for COERR ($N\cdot(h_{T})^{D}$)} & $125$ &  $200$ & $250$   \\
		\textbf{\# arms for CUCB}  & 121  & 487 & 991\\
		\hline
	\end{tabular}\\
\end{table}

\section{Conclusion}\label{sec:conclusion}
In this paper, we investigated the edge resource rental problem to facilitate the edge service provisioning in a shared edge system. An online decision-making policy, called Context-aware Online Edge Resource Rental (COERR), is designed for ASP to make appropriate edge resource rental decisions while learning the service demand pattern for each individual edge sites. COERR is developed based on the framework of contextual combinatorial multi-armed bandit, where ASP observes the context of SBSs and learns context-aware service demand to guide the resource at multiple edge sites. The proposed algorithm is easy to implement and guarantees provable asymptotically optimal performance.  However, there are still efforts left to be done to improve COERR. First, we currently use a static partition of context space. Considering dynamic partition may help improve the algorithm performance since it generates more appropriate hypercubes for learners in each time slot. Second, we currently only consider the edge resource rental for one ASP. Extending our algorithm to the multi-ASP scenario would be more beneficial in practice.

\bibliographystyle{IEEEtran}
\bibliography{reference}

\begin{thebibliography}{10}
\providecommand{\url}[1]{#1}
\csname url@samestyle\endcsname
\providecommand{\newblock}{\relax}
\providecommand{\bibinfo}[2]{#2}
\providecommand{\BIBentrySTDinterwordspacing}{\spaceskip=0pt\relax}
\providecommand{\BIBentryALTinterwordstretchfactor}{4}
\providecommand{\BIBentryALTinterwordspacing}{\spaceskip=\fontdimen2\font plus
\BIBentryALTinterwordstretchfactor\fontdimen3\font minus
  \fontdimen4\font\relax}
\providecommand{\BIBforeignlanguage}[2]{{%
\expandafter\ifx\csname l@#1\endcsname\relax
\typeout{** WARNING: IEEEtran.bst: No hyphenation pattern has been}%
\typeout{** loaded for the language `#1'. Using the pattern for}%
\typeout{** the default language instead.}%
\else
\language=\csname l@#1\endcsname
\fi
#2}}
\providecommand{\BIBdecl}{\relax}
\BIBdecl

\bibitem{dinh2013survey}
H.~T. Dinh, C.~Lee, D.~Niyato, and P.~Wang, ``A survey of mobile cloud
  computing: architecture, applications, and approaches,'' \emph{Wireless
  communications and mobile computing}, vol.~13, no.~18, pp. 1587--1611, 2013.

\bibitem{satyanarayanan2017emergence}
M.~Satyanarayanan, ``The emergence of edge computing,'' \emph{Computer},
  vol.~50, no.~1, pp. 30--39, 2017.

\bibitem{vaporio}
\BIBentryALTinterwordspacing
VaporIO. (2018) Introducing the kinetic edge: Extend your service to the
  kinetic edge. [Online]. Available: \url{https://www.vapor.io/}
\BIBentrySTDinterwordspacing

\bibitem{stan2018uptime}
M.~STANSBERRY, ``Uptime institute-data center industry survey,'' \url{
  https://uptimeinstitute.com/2018-data-center-industry-survey-results}.

\bibitem{iosup2008grid}
A.~Iosup, H.~Li, M.~Jan, S.~Anoep, C.~Dumitrescu, L.~Wolters, and D.~H. Epema,
  ``The grid workloads archive,'' \emph{Future Generation Computer Systems},
  vol.~24, no.~7, pp. 672--686, 2008.

\bibitem{mao2017mobile}
Y.~Mao, C.~You, J.~Zhang, K.~Huang, and K.~B. Letaief, ``A survey on mobile
  edge computing: The communication perspective,'' \emph{IEEE Communications
  Surveys \& Tutorials}, 2017.

\bibitem{saguna}
Intel. (2018) Saguna* and intel – using mobile edge computing to improve
  mobile network performance and profitability.

\bibitem{chen2017socially}
L.~Chen and J.~Xu, ``Socially trusted collaborative edge computing in ultra
  dense networks,'' in \emph{Proceedings of the Second ACM/IEEE Symposium on
  Edge Computing}.\hskip 1em plus 0.5em minus 0.4em\relax ACM, 2017, p.~9.

\bibitem{liu2018orchestrator}
Q.~Liu, S.~Huang, J.~Opadere, and T.~Han, ``An edge network orchestrator for
  mobile augmented reality,'' in \emph{International Conference on Computer
  Communications(INFOCOM)}.\hskip 1em plus 0.5em minus 0.4em\relax IEEE, 2018,
  pp. 1--9.

\bibitem{chen2002dynamic}
Y.~Chen, R.~H. Katz, and J.~D. Kubiatowicz, ``Dynamic replica placement for
  scalable content delivery,'' in \emph{International Workshop on Peer-to-Peer
  Systems}.\hskip 1em plus 0.5em minus 0.4em\relax Springer, 2002, pp.
  306--318.

\bibitem{shanmugam2013femtocaching}
K.~Shanmugam, N.~Golrezaei, A.~G. Dimakis, A.~F. Molisch, and G.~Caire,
  ``Femtocaching: Wireless content delivery through distributed caching
  helpers,'' \emph{IEEE Transactions on Information Theory}, vol.~59, no.~12,
  pp. 8402--8413, 2013.

\bibitem{zhao2018red}
T.~Zhao, I.-H. Hou, S.~Wang, and K.~Chan, ``Red/led: An asymptotically optimal
  and scalable online algorithm for service caching at the edge,'' \emph{IEEE
  Journal on Selected Areas in Communications}, vol.~36, no.~8, pp. 1857--1870,
  2018.

\bibitem{xu2018joint}
J.~Xu, L.~Chen, and P.~Zhou, ``Joint service caching and task offloading for
  mobile edge computing in dense networks,'' in \emph{International Conference
  on Computer Communications(INFOCOM)}.\hskip 1em plus 0.5em minus 0.4em\relax
  IEEE, 2018, pp. 1--9.

\bibitem{xie2018dynamic}
Q.~Xie, Q.~Wang, N.~Yu, H.~Huang, and X.~Jia, ``Dynamic service caching in
  mobile edge networks,'' in \emph{2018 IEEE 15th International Conference on
  Mobile Ad Hoc and Sensor Systems (MASS)}.\hskip 1em plus 0.5em minus
  0.4em\relax IEEE, 2018, pp. 73--79.

\bibitem{chen2018spatio}
L.~Chen, J.~Xu, S.~Ren, and P.~Zhou, ``Spatio-temporal edge service placement:
  A bandit learning approach,'' \emph{IEEE Transactions on Wireless
  Communications}, 2018.

\bibitem{lai1985asymptotically}
T.~L. Lai and H.~Robbins, ``Asymptotically efficient adaptive allocation
  rules,'' \emph{Advances in applied mathematics}, vol.~6, no.~1, pp. 4--22,
  1985.

\bibitem{auer2002finite}
P.~Auer, N.~Cesa-Bianchi, and P.~Fischer, ``Finite-time analysis of the
  multiarmed bandit problem,'' \emph{Machine learning}, vol.~47, no. 2-3, pp.
  235--256, 2002.

\bibitem{gai2012combinatorial}
Y.~Gai, B.~Krishnamachari, and R.~Jain, ``Combinatorial network optimization
  with unknown variables: Multi-armed bandits with linear rewards and
  individual observations,'' \emph{IEEE/ACM Transactions on Networking (TON)},
  vol.~20, no.~5, pp. 1466--1478, 2012.

\bibitem{li2010contextual}
L.~Li, W.~Chu, J.~Langford, and R.~E. Schapire, ``A contextual-bandit approach
  to personalized news article recommendation,'' in \emph{Proceedings of the
  19th international conference on World wide web}.\hskip 1em plus 0.5em minus
  0.4em\relax ACM, 2010, pp. 661--670.

\bibitem{tekin2015distributed}
C.~Tekin and M.~van~der Schaar, ``Distributed online learning via cooperative
  contextual bandits,'' \emph{IEEE Transactions on Signal Processing}, vol.~63,
  no.~14, pp. 3700--3714, 2015.

\bibitem{li2016contextual}
S.~Li, B.~Wang, S.~Zhang, and W.~Chen, ``Contextual combinatorial cascading
  bandits,'' in \emph{International Conference on Machine Learning}, 2016, pp.
  1245--1253.

\bibitem{muller2017context}
S.~M{\"u}ller, O.~Atan, M.~van~der Schaar, and A.~Klein, ``Context-aware
  proactive content caching with service differentiation in wireless
  networks,'' \emph{IEEE Transactions on Wireless Communications}, vol.~16,
  no.~2, pp. 1024--1036, 2017.

\bibitem{chen2018contextual}
L.~Chen, J.~Xu, and Z.~Lu, ``Contextual combinatorial multi-armed bandits with
  volatile arms and submodular reward.'' in \emph{Conference on Neural
  Information Processing Systems (NeurIPS)}, 2018, pp. 1--10.

\bibitem{chen2018computation}
L.~Chen, S.~Zhou, and J.~Xu, ``Computation peer offloading for
  energy-constrained mobile edge computing in small-cell networks,''
  \emph{IEEE/ACM Transactions on Networking}, vol.~26, no.~4, pp. 1619--1632,
  2018.

\bibitem{mao2016dynamic}
Y.~Mao, J.~Zhang, and K.~B. Letaief, ``Dynamic computation offloading for
  mobile-edge computing with energy harvesting devices,'' \emph{IEEE Journal on
  Selected Areas in Communications}, vol.~34, no.~12, pp. 3590--3605, 2016.

\bibitem{sonmez2011matching}
T.~S{\"o}nmez and M.~U. {\"U}nver, ``Matching, allocation, and exchange of
  discrete resources,'' in \emph{Handbook of social Economics}.\hskip 1em plus
  0.5em minus 0.4em\relax Elsevier, 2011, vol.~1, pp. 781--852.

\bibitem{chandrasekhar2009spectrum}
V.~Chandrasekhar and J.~G. Andrews, ``Spectrum allocation in tiered cellular
  networks,'' \emph{IEEE Transactions on Communications}, vol.~57, no.~10, pp.
  3059--3068, 2009.

\bibitem{langford2008epoch}
J.~Langford and T.~Zhang, ``The epoch-greedy algorithm for multi-armed bandits
  with side information,'' in \emph{Advances in Neural Information Processing
  Systems (NeurIPS)}, 2008, pp. 817--824.

\bibitem{lin2009global}
Y.~Lin and L.~Schrage, ``The global solver in the lindo api,''
  \emph{Optimization Methods \& Software}, vol.~24, no. 4-5, pp. 657--668,
  2009.

\bibitem{cplex2009v12}
IBM, ``V12.1: User’s manual for cplex,'' \emph{International Business
  Machines Corporation}, vol.~46, no.~53, p. 157, 2009.

\bibitem{bergman1999recursive}
N.~Bergman, ``Recursive bayesian estimation,'' \emph{Department of Electrical
  Engineering, Link{\"o}ping University, Link{\"o}ping Studies in Science and
  Technology. Doctoral dissertation}, vol. 579, p.~11, 1999.

\bibitem{lee1981recursive}
D.~Lee, M.~Morf, and B.~Friedlander, ``Recursive least squares ladder
  estimation algorithms,'' \emph{IEEE Transactions on Acoustics, Speech, and
  Signal Processing}, vol.~29, no.~3, pp. 627--641, 1981.

\bibitem{valiant2013probably}
L.~Valiant, \emph{Probably Approximately Correct: Nature{\~O}s Algorithms for
  Learning and Prospering in a Complex World}.\hskip 1em plus 0.5em minus
  0.4em\relax Basic Books (AZ), 2013.

\bibitem{onlineappendix}
\BIBentryALTinterwordspacing
L.~Chen and J.~Xu. Online appendix: Budget-constrained edge service
  provisioning with demand estimation via bandit learning. [Online]. Available:
  \url{https://www.dropbox.com/sh/yst3olzmcapx1x8/AAA0QT5wXvVNTS1_lWwSGPl9a?dl=0}
\BIBentrySTDinterwordspacing

\bibitem{hoeffding1963probability}
W.~Hoeffding, ``Probability inequalities for sums of bounded random
  variables,'' \emph{Journal of the American statistical association}, vol.~58,
  no. 301, pp. 13--30, 1963.

\bibitem{cesa1997use}
N.~Cesa-Bianchi, Y.~Freund, D.~Haussler, D.~P. Helmbold, R.~E. Schapire, and
  M.~K. Warmuth, ``How to use expert advice,'' \emph{Journal of the ACM
  (JACM)}, vol.~44, no.~3, pp. 427--485, 1997.

\bibitem{tekin2017adaptive}
C.~Tekin, J.~Yoon, and M.~Van Der~Schaar, ``Adaptive ensemble learning with
  confidence bounds,'' \emph{IEEE Transactions on Signal Processing}, vol.~65,
  no.~4, pp. 888--903, 2017.

\bibitem{martello1990knapsack}
S.~Martello, ``Knapsack problems: algorithms and computer implementations,''
  \emph{Wiley-Interscience series in discrete mathematics and optimiza tion},
  1990.

\bibitem{bettinelli2017branch}
A.~Bettinelli, V.~Cacchiani, and E.~Malaguti, ``A branch-and-bound algorithm
  for the knapsack problem with conflict graph,'' \emph{INFORMS Journal on
  Computing}, vol.~29, no.~3, pp. 457--473, 2017.

\bibitem{pferschy2017approximation}
U.~Pferschy and J.~Schauer, ``Approximation of knapsack problems with conflict
  and forcing graphs,'' \emph{Journal of Combinatorial Optimization}, vol.~33,
  no.~4, pp. 1300--1323, 2017.

\bibitem{chen2013combinatorial}
W.~Chen, Y.~Wang, and Y.~Yuan, ``Combinatorial multi-armed bandit: General
  framework and applications,'' in \emph{International Conference on Machine
  Learning}, 2013, pp. 151--159.

\bibitem{chu2011contextual}
W.~Chu, L.~Li, L.~Reyzin, and R.~Schapire, ``Contextual bandits with linear
  payoff functions,'' in \emph{Proceedings of the Fourteenth International
  Conference on Artificial Intelligence and Statistics}, 2011, pp. 208--214.

\end{thebibliography}

\newpage
\appendices

\section{Proof of Lemma \ref{lemma:R_exploit}}\label{proof:lemma:exploit}
\begin{proof}
Before proceeding, we first define several auxiliary variables: for a hypercube $p$ maintained by SBS $n$, we define $\bar{\mu}_n(p) = \sup_{x \in p} \mu_n(x)$ and $\ubar{\mu}_n(p) = \inf_{x\in p} \mu_n(x)$ be the best and worst expected demand of SBS $n$ over all contexts in hypercube $p$. Let 
\begin{align*}
	\bar{\bm{\mu}}(\p^t) = \{\bar{\mu}_1(p^t_1),\bar{\mu}_2(p^t_2),\dots,\bar{\mu}_N(p^t_N)\}\\
	\ubar{\bm{\mu}}(\p^t) = \{\ubar{\mu}_1(p^t_1),\ubar{\mu}_2(p^t_2),\dots,\ubar{\mu}_N(p^t_N)\}
\end{align*}
In some steps of the proofs, we need compare the service demands at different positions in a hypercube. As a point of reference, we define the context at the (geometrical) center of a hypercube $p$ as $\dot{x}(p)$. Let $\dot{\bm{\mu}}(\p^t) = \{\mu_1(\dot{x}(p^t_1)),\mu_2(\dot{x}(p^t_2)),\dots,\allowbreak\mu_N(\dot{x}(p^t_N))\}$. We also define the rental decision $\dot{\bm{f}}(\p^t)$ which is derived based on the expected service demand $\dot{\bm{\mu}}(\p^t)$ by solving the following problem in time slot $t$:
\begin{align}
	\dot{\bm{f}}(\p^t) = \argmax\nolimits_{\f\in\mathcal{F}} ~ U^t(\f; \dot{\bm{\mu}}(\p^t)) \quad \text{s.t.} ~ \eqref{cstr:sub_c1},\eqref{cstr:sub_c2}, \eqref{cstr:sub_c3}
\end{align}
The rental decision $\dot{\bm{f}}(\p^t)$ is used to identify the bad rental decisions when the hypercubes of contexts $\x^t$ is $\p^t$. Let
\begin{align}
	\mathcal{L}(\p^t) = \left\{\f\in\mathcal{F} \mid U^t(\dot{\f}(\p^t); \ubar{\bm{\mu}}(\p^t)) - U^t(\f; \bar{\bm{\mu}}(\p^t)) \geq At^{\theta}\right\}
\end{align}
be the set of suboptimal rental decisions when the SBSs' contexts belong to $\p^t$. The parameter $A>0$ and $\theta<0$ are only used in the regret analysis. We call a rental decision $\f\in\mathcal{L}(\p^t)$ \emph{suboptimal} for $\p^t$, since the ASP utility achieved by the rental decision $\dot{\f}(\p^t)$ is at least an amount $At^\theta$ higher than that achieved by the rental decision $\f\in\mathcal{L}(\p^t)$. We call the rental decisions in $\mathcal{F}\backslash\mathcal{L}(\p^t)$ \emph{near-optimal} for $\p^t$. Then, the regret of $R_\text{exploit}(T)$ can be divided into the following two summands:
\begin{align}
	\mathbb{E}[R_\text{exploit}(T)] = \mathbb{E}[R_\text{s}(T)] + \mathbb{E}[R_\text{n}(T)]
\end{align}
where the term $\mathbb{E}[R_\text{s}(T)]$ is the regret due to the suboptimal rental decision in exploitation and $\mathbb{E}[R_\text{n}(T)]$ is the regret due to the near-optimal rental decision in exploitation. In the following, we will show that each of the two summands is bounded. We first give the bound of $\mathbb{E}[R_\text{s}(T)]$ in Lemma \ref{lemma:E_Rs}. 
\begin{lemma}\label{lemma:E_Rs}
(Bound of $\mathbb{E}[R_\text{s}(T)]$.) Given the input parameters $h_T$ and $K(t)$, if the H\"{o}lder condition holds true and the additional condition $2H(t) + 2d^{\max}NLD^{\frac{\alpha}{2}}h_{T}^{\alpha} \leq At^\theta$ is satisfied with $H(t)>0$ for all $t$, then $\mathbb{E}[R_\text{s}(T)]$ is bounded by
\begin{align*}
	\mathbb{E}[R_\text{s}(T) \leq \frac{2|\mathcal{F}|B\lambda^{\max}d^{\max}}{w^{\min}}\sum_{t=1}^T\sum_{n\in\mathcal{N}}\sigma_n\left(\frac{H(t)}{d^{\max}N},K(t)\right)
\end{align*}
\end{lemma}
\begin{proof}
	Let $W(t) = \left\{\mathcal{U}^t = \emptyset \right\}$ be the event that the algorithm enters the exploitation phase. By the definition of $\mathcal{U}^t$, we will have $C^t_n(p^t_n)>K(t)$ for all $\forall n$. Let $V_{f}(t)$ be the event that rental decision $\f$ is taken in time slot $t$. Then, it holds that
	\begin{align}
		R_s(T) = \sum_{t=1}^T\sum_{\f\in\mathcal{L}(\p^t)} I_{\{V_{f},W(t)\}}\times \left(U^t(\f^{*,t};\bm{\lambda}^t)-U^t(\f;\bm{\lambda}^t)\right)
	\end{align}
	In each of the summands, the utility loss is considered due to taking a suboptimal decision $\f\in\mathcal{L}(\p^t)$ instead of the optimal Oracle decision $\f^{*,t}$. Since the maximum utility loss at a SBS is bounded by $\lambda^{\max}d^{\max}$, and the maximum number of SBSs that hosting the edge service is $B/w^{\min}$, we have 
	\begin{align}
		R_\text{s}(T) \leq \frac{B}{w^{\min}}\lambda^{\max}d^{\max}\sum_{t=1}^{T}\sum_{\f\in\mathcal{L}(\p^t)} I_{\left\{V_{f}(t),W(t)\right\}},
	\end{align}
	Taking the expectation, the regret due to suboptimal decisions is bounded by
	\begin{align*}
		\mathbb{E}\left[R_\text{s}(T)\right] & \leq \frac{B}{w^{\min}}\lambda^{\max}d^{\max}\sum_{t=1}^{T}\sum_{\f\in\mathcal{L}(\p^t)} \mathbb{E}\left[I_{\left\{V_{f}(t),W(t)\right\}}\right]\\
		& = \frac{B}{w^{\min}}\lambda^{\max}d^{\max}\sum_{t=1}^{T}\sum_{\f\in\mathcal{L}(\p^t)} \Pr\left\{V_{f}(t),W(t)\right\}
	\end{align*}
	Based on the algorithm design, we know that if a rental decision $\f$ is taken in exploration (i.e., the event $V_{f}(t)$), we must have $U^t(\f,\hat{\bm{\lambda}}^t) \geq U^t(\dot{\f}(\p^t),\hat{\bm{\lambda}}^t)$, Thus. we have
	\begin{align}\label{eq:Pr_VW}
		\Pr\left\{V_{f}(t),W(t)\right\} \leq \Pr\left\{U^t(\f,\hat{\bm{\lambda}}^t) \geq U^t(\dot{\f}(\p^t),\hat{\bm{\lambda}}^t),  W(t)\right\}
	\end{align}
	The right-hand side of \eqref{eq:Pr_VW} implies at least one of the three following events with any $H(t)>0$:
	\begin{align*}
		E_1 = \left\{U^t(\f,\hat{\bm{\lambda}}^t) \geq U^t(\f,\bar{\bm{\mu}}(\p^t)) + H(t), W(t)\right\},
	\end{align*} 
	\begin{align*}
		E_2 = \left\{U^t(\dot{\f}(\p^t),\hat{\bm{\lambda}}^t) \leq U^t(\dot{\f}(\p^t),\ubar{\bm{\mu}}(\p^t)) - H(t), W(t)\right\},
	\end{align*} 
	\begin{align*}
		E_3 = 	&\left\{U^t(\f,\hat{\bm{\lambda}}^t) \geq U^t(\dot{\f}(\p^t),\hat{\bm{\lambda}}^t),\right.\\
				&\left. ~~ U^t(\f,\hat{\bm{\lambda}}^t) < U^t(\f,\bar{\bm{\mu}}(\p^t)) + H(t),\right.\\
				& \left.~~ U^t(\dot{\f}(\p^t),\hat{\bm{\lambda}}^t) > U^t(\dot{\f}(\p^t),\ubar{\bm{\mu}}(\p^t) - H(t), W(t)\right\}.
	\end{align*}
	Hence, we have for the original event in \eqref{eq:Pr_VW}
	\begin{align}
		\left\{U^t(\f,\hat{\bm{\lambda}}^t) \geq U^t(\dot{\f}(\p^t),\hat{\bm{\lambda}}^t),  W(t)\right\} \subseteq \left\{ E_1 \cup E_2 \cup E_3 \right\}
	\end{align}
	The probability of these three events $E_1$, $E_2$, and $E_3$ will be bounded separately. Let us start with $E_1$. Recall that the best expected service demand for the hypercube $p$ at SBS $n$ is $\bar{\mu}_n(p) = \sup_{x \in p} \mu_n(x)$ and we must have $\mathbb{E}_{x \in p}[\mu_n(x)] = \mu_n(p) \leq \sup_{x\in p} \mu_n(x) = \bar{\mu}_n(p)$ in Assumption \ref{ass:PAC}, we will have:
	\begin{align*}
		\Pr\{E_1\} & = \Pr\left\{U^t(\f,\hat{\bm{\lambda}}^t)\geq U^t(\f,\bar{\bm{\mu}}(\p^t))+H(t), W(t)\right\}\\
		& = \Pr\left\{\sum_{n\in\mathcal{N}} \hat{\lambda}_n(p^t_n)\Delta^t_n(f^t_n) \geq \sum_{n\in\mathcal{N}} \bar{\mu}_n(p^t_n)\Delta^t_n(f^t_n)  + H(t), W(t)\right\}\\
		& \leq \Pr\left\{\sum_{n\in\mathcal{N}} \hat{\lambda}_n(p^t_n)\Delta^t_n(f^t_n) \geq \sum_{n\in\mathcal{N}} \mu_n(p^t_n)\Delta^t_n(f^t_n) + H(t), W(t)\right\}\\
		& \leq \sum_{n\in\mathcal{N}} \Pr\left\{\hat{\lambda}_n(p^t_n)\Delta^t_n(f^t_n) \geq \mu_n(p^t_n)\Delta^t_n(f^t_n) + \frac{H(t)}{N}, W(t)\right\}\\
		& \leq \sum_{n\in\mathcal{N}} \Pr\left\{ \hat{\lambda}_n(p^t_n) \geq \mu_n(p^t_n) + \frac{H(t)}{d^{\max}N}, W(t)\right\}\\
	\end{align*}
	Considering the PAC condition of the estimators $\Pr\left\{\hat{\lambda}_n(p^t_n) - \mu_n(p^t_n) \geq \epsilon\right\} = \sigma_n(\epsilon, C^t_n(p^t_n))$, we have
	\begin{align*}
		\Pr\{E_1\} & \leq \sum_{n\in\mathcal{N}} \sigma_n\left(\frac{H(t)}{d^{\max}N},C^t_n(p^t_n)\right)
	\end{align*}
	Analogously, it can be proven for event $E_2$ that
	\begin{align*}
		\Pr\{E_2\} & \leq \sum_{n\in\mathcal{N}} \sigma_n\left(\frac{H(t)}{d^{\max}N},C^t_n(p^t_n)\right)
	\end{align*}
	
	To bound event $E_3$, we first make some additional definition. First, we rewrite the service demand estimation $\hat{\lambda}_n(p^t_n)$ as below:
	\begin{align*}
		 \hat{\lambda}_n(p^t_n) = \mu_n(x^t_n) + \zeta^t_n
	\end{align*}
	where $\zeta^t_n$ denotes the deviation of estimation $\hat{\lambda}_n(p^t_n)$ from the expected service demand of SBS $n$ with context $x^t_n$. Additional, we define the best and worst context in the hypercube $p^t_n$ as $x^\text{best}(p^t_n) = \argmax_{x\in p^t_n} \mu(x)$ and $x^\text{worst}(p^t_n) = \argmin_{x\in p^t_n} \mu(x)$, respectively. Let 
	\begin{align*}
		\lambda_n^\text{best}(p^t_n) = \mu_n(x^\text{best}(p^t_n)) + \zeta^t_n\\
		\lambda_n^\text{worst}(p^t_n) = \mu_n(x^\text{worst}(p^t_n)) + \zeta^t_n\\
	\end{align*}
	By the H\"{o}lder condition, it can be shown that    
	\begin{align*}
		\lambda^\text{best}_n(p^t_n) - \hat{\lambda}(p^t_n) = \mu_n(x^\text{best}(p^t_n)) - \mu_n(x^t_n)\leq LD^{\frac{\alpha}{2}}h_{T}^{-\alpha}\\
		\hat{\lambda}(p^t_n) - \lambda_n^\text{worst}(p^t_n)= \mu_n(x^t_n) - \mu_n(x^\text{worst}(p^t_n)) \leq LD^{\frac{\alpha}{2}}h_{T}^{-\alpha}
	\end{align*}
	Applying the above results, we will have
	\begin{align}
		& U^t(\f,\bm{\lambda}^\text{best}(\p^t)) - U^t(\f,\hat{\bm{\lambda}}^t) \leq d^{\max}\sum_{n\in\mathcal{N}} LD^{\frac{\alpha}{2}}h_{T}^{-\alpha} \label{eq:temp1}\\
		& U^t(\dot{\f}(\p^t),\hat{\bm{\lambda}}^t) - U^t(\dot{\f}(\p^t),\bm{\lambda}^\text{worst}(\p^t)) \leq d^{\max}\sum_{n\in\mathcal{N}} LD^{\frac{\alpha}{2}}h_{T}^{-\alpha}  \label{eq:temp2}
	\end{align}
	where $\bm{\lambda}^\text{best}(\p^t) = \{	\lambda^\text{best}_1(p^t_1), \dots, \lambda^\text{best}_N(p^t_N) \}$ and $\bm{\lambda}^\text{worst}(\p^t) = \{\lambda^\text{worst}_1(p^t_1), \dots,\lambda^\text{worst}_N(p^t_N) \}$.
	By the definition of $\bm{\lambda}^\text{best}(\p^t)$ and $\bm{\lambda}^\text{worst}(\p^t) $, it holds for the first component of $E_3$ that
	\begin{align*}
		E_{3.1} = \left\{U^t(\f,\hat{\bm{\lambda}}^t) \geq U^t(\dot{\f}(\p^t),\hat{\bm{\lambda}}^t),\right\} \subseteq \left\{U^t(\f,\bm{\lambda}^\text{best}(\p^t)) \geq U^t(\dot{\f}(\p^t),\bm{\lambda}^\text{worst}(\p^t))\right\}
	\end{align*}
	For the second component, using \eqref{eq:temp1}, we have
	\begin{align*}
		E_{3.2} = &\left\{ U^t(\f,\hat{\bm{\lambda}}^t) < U^t(\f,\bar{\bm{\lambda}}(\p^t)) + H(t),\right\}\\
		\subseteq &\left\{ U^t(\f, \bm{\lambda}^\text{best}(\p^t)) - d^{\max}\sum_{n\in\mathcal{N}} LD^{\frac{\alpha}{2}}h_{T}^{-\alpha} < U^t(\f,\bar{\bm{\lambda}}(\p^t)) + H(t)\right\}\\
		= & \left\{ U^t(\f,\bm{\lambda}^\text{best}(\p^t))  < U^t(\f,\bar{\bm{\lambda}}(\p^t)) +  d^{\max}\sum_{n\in\mathcal{N}} LD^{\frac{\alpha}{2}}h_{T}^{-\alpha} + H(t)\right\}
	\end{align*}
	For the third component, using \eqref{eq:temp2}, we have
	\begin{align*}
		E_{3.3} = &\left\{U^t(\dot{\f}(\p^t),\hat{\bm{\lambda}}^t) > U^t(\dot{\f}(\p^t),\ubar{\bm{\lambda}}(\p^t)) - H(t)\right\}\\
		\subseteq & \left\{U^t(\dot{\f}(\p^t),\bm{\lambda}^\text{worst}(\p^t)) + d^{\max}\sum_{n\in\mathcal{N}} LD^{\frac{\alpha}{2}}h_{T}^{-\alpha} > U^t(\dot{\f}(\p^t),\ubar{\bm{\lambda}}(\p^t)) - H(t)\right\} \\
		= & \left\{U^t(\dot{\f}(\p^t),\bm{\lambda}^\text{worst}(\p^t)) > U^t(\dot{\f}(\p^t),\ubar{\bm{\lambda}}(\p^t))- d^{\max}\sum_{n\in\mathcal{N}} LD^{\frac{\alpha}{2}}h_{T}^{-\alpha} - H(t)\right\}
	\end{align*}
	We want to find a condition under which the probability of $E_3$ is zero. To this end, we design the parameter $H(t)$, such that 
	\begin{align}\label{eq:condition_H}
		2H(t) + 2d^{\max}NLD^{\frac{\alpha}{2}}h_{T}^{-\alpha} \leq At^\theta
	\end{align}
	Since $\f\in\mathcal{L}(\p^t)$, we have $U^t(\dot{\f}(\p^t); \ubar{\bm{\lambda}}(\p^t)) - U^t(\f^t; \bar{\bm{\lambda}}(\p^t)) \geq At^\theta$, which together with \eqref{eq:condition_H}, implies that:
	\begin{align*}
		U^t(\dot{\f}(\p^t); \ubar{\bm{\lambda}}(\p^t)) - U^t(\f; \bar{\bm{\lambda}}(\p^t)) - 2H(t) - 2d^{\max}NLD^{\frac{\alpha}{2}}h_{T}^{-\alpha}  \geq 0
	\end{align*}
	Rewriting yields:
	\begin{align}\label{eq:temp3}
		U^t(\dot{\f}(\p^t); \ubar{\bm{\lambda}}(\p^t)) - d^{\max}NLD^{\frac{\alpha}{2}}h_{T}^{-\alpha} - H(t)  \geq U^t(\f; \bar{\bm{\lambda}}(\p^t))  + d^{\max}NLD^{\frac{\alpha}{2}}h_{T}^{-\alpha} + H(t)
	\end{align}
	If \eqref{eq:temp3} holds true, the three components of $E_3$ cannot be satisfied at the same time: combining $E_{3.2}$ and $E_{3.3}$ with \eqref{eq:temp3} yields $U^t(\f,\bm{\lambda}^\text{best}(\p^t)) < U^t(\dot{\f}(\p^t),\bm{\lambda}^\text{worst}(\p^t))$, which contradicts the $E_{3.1}$. Therefore, under the condition \eqref{eq:condition_H}, $\Pr\{E_3\} = 0$.
	
	To sum up, under condition \eqref{eq:condition_H}, the probability $\Pr\{V_f(t),W(t)\}$ is bounded by
	\begin{align*}
		\Pr\{V_f(t),W(t)\} &\leq \Pr\{E_1\cup E_2\cup E_3\}\\
		& \leq \Pr\{E_1\} + \Pr\{E_2\} + \Pr\{E_3\} \\
		& \leq 2 \sum_{n\in\mathcal{N}}\sigma_n\left(\frac{H(t)}{d^{\max}N},C^t_n(p^t_n)\right)\\
		& \leq 2 \sum_{n\in\mathcal{N}}\sigma_n\left(\frac{H(t)}{d^{\max}N}, K(t)\right)
	\end{align*}
	where the last inequality is due to the estimator property $\frac{\partial\sigma(\epsilon,C^t_n(p))}{\partial C^t_n(p)}\leq 0$ in Assumption \ref{ass:PAC} and the fact that in the exploration phase an arbitrary counter satisfies $C^t_n(p^t_n)\geq K(t)$. 
	Therefore, the regret bound for $\mathbb{E}[R_\text{s}(T)]$ is 
	\begin{align*}
		\mathbb{E}[R_\text{s}(T)] &\leq \frac{B}{w^{\min}}\lambda^{\max}d^{\max}\sum_{t=1}^T\sum_{f\in\mathcal{L}(\p^t)}\Pr\{V_{f}(t),W(t)\}\\
		&\leq \frac{2|\mathcal{F}|B\lambda^{\max}d^{\max}}{w^{\min}}\sum_{t=1}^T\sum_{n\in\mathcal{N}}\sigma_n\left(\frac{H(t)}{d^{\max}N},K(t)\right)
	\end{align*}
\end{proof}

Next we bound the regret due to choosing near-optimal rental decisions, which is given in the Lemma below.
\begin{lemma}\label{lemma:E_Rn}
(Bound for $\mathbb{E}[R_\text{n}(T)]$). Given the input parameters $h_T$ and $K(t)$, if the H\"{o}lder condition holds true, the regret $\mathbb{E}[R_\text{n}(T)]$ is bounded by
	\begin{align*}
		\mathbb{E}[R_\text{n}(T)] \leq 3d^{\max}NLD^{\frac{\alpha}{2}}h_{T}^{-\alpha}T + AT^{\theta+1}
	\end{align*}
	
\end{lemma}
\begin{proof}
	For an arbitrary time slot $t$, let the event $W(t)$ denotes that the algorithm enters the exploration phase. Let $Q(t)$ be the event that an near optimal rental decision $\f^t \in \mathcal{F} \backslash \mathcal{L}(\p^t)$ is taken in time slot $t$. The loss due to the near-optimal subsets can be written as:
	\begin{align*}
	R_\text{n}(T) = \sum_{t=1}^{T} I_{\{W(t), Q(t)\}} \times \left(U^t(\f^{*,t};\bm{\lambda}^t)-U^t(\f^t;\bm{\lambda}^t)\right)
	\end{align*}
	Taking the expectation of $R_\text{n}(T)$, by the definition of conditional expectation, we have:
	\begin{align*}
	\mathbb{E}[R_\text{n}(T)] & = \sum_{t=1}^{T} \Pr\{W(t),Q(t)\} \cdot \mathbb{E}\left[U^t(\f^{*,t};\bm{\lambda}^t)-U^t(\f^t;\bm{\lambda}^t) \mid  W(t), Q(t) \right]\\
	& \leq  \sum_{t=1}^{T}\mathbb{E}\left[U^t(\f^{*,t};\bm{\lambda}^t)-U^t(\f^t;\bm{\lambda}^t) \mid W(t),Q(t) \right].
	\end{align*}
	Since $\f^t \in \mathcal{F} \backslash \mathcal{L}(\p^t)$, it holds that 
	\begin{align*}
		U^t(\dot{\f}(\p^t); \ubar{\bm{\lambda}}(\p^t)) - U^t(\f^t; \bar{\bm{\lambda}}(\p^t)) < At^{\theta}.
	\end{align*}
	To bound the regret, we have to give an upper bound on
	\begin{align*}
		\sum_{t=1}^{T}\mathbb{E}\left[U^t(\f^{*,t};\bm{\lambda}^t)-U^t(\f^t;\bm{\lambda}^t) \mid W(t), Q(t) \right] = \sum_{t=1}^{T} U^t(\f^{*,t};\bm{\mu}^t)-U^t(\f^t;\bm{\mu}^t).
	\end{align*}
	Applying H\"{o}lder condition several times yields
	\begin{align*}
		&U^t(\f^{*,t};\bm{\mu}^t)-U^t(\f^t;\bm{\mu}^t) \\
		\leq & 	U^t(\f^{*,t};\dot{\bm{\mu}}(\p^t)) + d^{\max}NLD^{\frac{\alpha}{2}}h_{T}^{-\alpha} - U^t(\f^t;\bm{\mu}^t) \\
		\leq & U^t(\dot{\f}(\p^t);\dot{\bm{\mu}}(\p^t)) + d^{\max}NLD^{\frac{\alpha}{2}}h_{T}^{-\alpha} - U^t(\f^t;\bm{\mu}^t) \\
		\leq & 	U^t(\dot{\f}(\p^t);\ubar{\bm{\mu}}(\p^t)) + 2d^{\max}NLD^{\frac{\alpha}{2}}h_{T}^{-\alpha} - U^t(\f^t;\bm{\mu}^t) \\
		\leq & 	U^t(\dot{\f}(\p^t);\ubar{\bm{\mu}}(\p^t)) + 3d^{\max}NLD^{\frac{\alpha}{2}}h_{T}^{-\alpha} - U^t(\f^t;\bar{\bm{\mu}}(\p^t))\\
		\leq & 3d^{\max}NLD^{\frac{\alpha}{2}}h_{T}^{-\alpha}  + At^\theta
	\end{align*}
	
	Therefore, $\mathbb{E}[R_\text{n}(T)]$ can be bounded by
	\begin{align*}
		\mathbb{E}[R_\text{n}(T)] & \leq \sum_{t=1}^{T}\left( 3d^{\max}NLD^{\frac{\alpha}{2}}h_{T}^{-\alpha}  + At^\theta\right)\\
		& \leq 3d^{\max}NLD^{\frac{\alpha}{2}}h_{T}^{-\alpha}T + AT^{\theta+1}
	\end{align*}
\end{proof}
Then the regret $\mathbb{E}[R_\text{exploit}(T)]$ is therefore bounded by 
\begin{align*}
	\mathbb{E}[R_\text{exploit}(T)] \leq \frac{2|\mathcal{F}|B\lambda^{\max}d^{\max}}{w^{\min}}\sum_{t=1}^T\sum_{n\in\mathcal{N}}\sigma_n\left(\frac{H(t)}{d^{\max}N},K(t)\right) + 3d^{\max}NLD^{\frac{\alpha}{2}}h_{T}^{-\alpha}T + AT^{\theta+1}.
\end{align*}
\end{proof}

\section{Proof of Theorem \ref{theo:regret_bound}}\label{proof:theo:regret_bound}
\begin{proof}
	Let $K(t)= t^z\log(t), 0<z<1$ and $h_T = \lceil T^\gamma \rceil, 0<\gamma,<\frac{1}{D}$, then $\mathbb{E}[R_\text{explore}(T)]$ in Lemma \ref{lemma:R_explore} can be rewrite as 
	\begin{align*}
	\mathbb{E}[R_\text{explore}(T)] & \leq \frac{NB\lambda^{\max}d^{\max}}{w^{\min}}(h_T)^D\lceil K(T)\rceil\\
	& =  \frac{NB\lambda^{\max}d^{\max}}{w^{\min}}\lceil T^\gamma \rceil^D\lceil T^z\log(T)\rceil
	\end{align*}
	Since $\lceil T^\gamma \rceil^D \leq (2T^\gamma)^D$, it holds that
	\begin{align*}
	\mathbb{E}[R_\text{explore}(T)] & \leq   \frac{NB\lambda^{\max}d^{\max}}{w^{\min}}2^DT^{\gamma D} (T^z\log(T)+1)\\
	& =  \frac{NB\lambda^{\max}d^{\max}}{w^{\min}}2^D (T^{z+\gamma D}\log(T)+T^{\gamma D})
	\end{align*}
	Consider the Lemma \ref{lemma:R_exploit}, we let $H(t)=N\lambda^{\max}d^{\max}t^{-z/2}$. Given $\sigma_n(\epsilon,C^t_n(p^t_n)) = \exp\left({-\dfrac{2C^t_n(p^t_n)\epsilon^2}{(\lambda^{\max})^2}}\right)$, the first term in \eqref{eq:R_exploit} can be written as
	\begin{align*}
	\frac{2|\mathcal{F}|B\lambda^{\max}d^{\max}}{w^{\min}}\sum_{t=1}^T\sum_{n\in\mathcal{N}}\sigma_n\left(\frac{H(t)}{d^{\max}N},K(t)\right) &= \frac{2|\mathcal{F}|B\lambda^{\max}d^{\max}}{w^{\min}}N\sum_{t=1}^{T} \exp\left(-\frac{2K(t)H^2(t)}{(\lambda^{\max}d^{\max}N)^2}\right)\\
	& = \frac{2N|\mathcal{F}|B\lambda^{\max}d^{\max}}{w^{\min}}\sum_{t=1}^{T} \exp\left(-2\log(t)\right)\\
	& =  \frac{2N|\mathcal{F}|B\lambda^{\max}d^{\max}}{w^{\min}}\sum_{t=1}^{T} t^{-2}\\
	&  \leq \frac{N|\mathcal{F}|B\lambda^{\max}d^{\max}}{w^{\min}} \cdot 2\sum_{t=1}^{\infty} t^{-2}\\
	& \leq \frac{N|\mathcal{F}|B\lambda^{\max}d^{\max}}{w^{\min}} \frac{\pi^2}{3}
	\end{align*}
	Note that in the second term of \eqref{eq:R_exploit}, $h_t^{-\alpha} = \lceil T^\gamma\rceil^{-\alpha} \leq T^{-\gamma\alpha}$ .Now the total regret is bounded by 
	\begin{align*}
	\mathbb{E}[R(T)] \leq &\frac{NB\lambda^{\max}d^{\max}}{w^{\min}}2^D (T^{z+\gamma D}\log(T) +T^{\gamma D}) \\ & + \frac{N|\mathcal{F}|B\lambda^{\max}d^{\max}\pi^2}{3w^{\min}} +  3d^{\max}NLD^{\frac{\alpha}{2}}T^{1-\gamma\alpha} + AT^{\theta+1}
	\end{align*}
	
	The summands above contribute to the regret with leading orders $O(T^{z+\gamma D}\log(T))$, $O(T^{1-\gamma \alpha})$, and $O(T^{\theta+1})$. In order to balance the leading orders, we let $z=\frac{2\alpha}{3\alpha+D}\in(0,1)$, $\gamma = \frac{z}{2\alpha}\in(0,\frac{1}{D})$, $\theta = - \frac{z}{2}$, and $A=2N\lambda^{\max}d^{\max} + 2d^{\max}NLD^{\alpha/2}$. With these parameters, the conditions in Lemma \ref{lemma:R_exploit} are satisfied. Now the regret $\mathbb{E}[R(T)]$ reduces to
	\begin{align*}
	\mathbb{E}[R(T)] \leq &\frac{NB\lambda^{\max}d^{\max}}{w^{\min}}2^D (T^{\frac{2\alpha+D}{3\alpha+D}}\log(T) +T^{\frac{D}{3\alpha+D}}) \\ & + \frac{N|\mathcal{F}|B\lambda^{\max}d^{\max}\pi^2}{3w^{\min}} +  3d^{\max}NLD^{\frac{\alpha}{2}}T^{\frac{2\alpha+D}{3\alpha+D}} + AT^{\frac{2\alpha+D}{3\alpha+D}}
	\end{align*}
	The proof is completed.
\end{proof}

\section{Proof of Theorem \ref{theo:delat_regret_bound}}\label{proof:theo:delat_regret_bound}
\begin{proof}

The proof of Theorem \ref{theo:delat_regret_bound} is similar to that of Theorem \ref{theo:regret_bound} and hence we only provide a sketch of proof for Theorem \ref{theo:delat_regret_bound}. The expect $\delta$-regret is also divided into two parts: 
\begin{align*}
	\mathbb{E}[R^\delta(T)] = \mathbb{E}[ R^\delta_\text{explore}(T) + R^\delta_\text{exploit}(T)]
\end{align*}
The approximation algorithm does not have much influences in bounding $\mathbb{E}[ R^\delta_\text{explore}(T)]$ since the worst-case utility loss $\lambda^{\max}d^{\max}$ is used to provide a upper bound of the regret incurred by a wrong rental decision at a SBS. According to the definition of $\delta$-regret, the worst-case utility loss becomes $\frac{1}{\delta}\lambda^{\max}d^{\max}$. By following the steps in proof of Lemma \ref{lemma:R_explore}, $\mathbb{E}[ R^\delta_\text{explore}(T)]$ is bounded by:
\begin{align*}
\mathbb{E}[R^\delta_\text{explore}(T)] \leq \frac{NB\lambda^{\max}d^{\max}}{\delta w^{\min}}(h_T)^D\lceil K(T)\rceil
\end{align*} 
Letting $h_T = \lceil T^{\frac{1}{3\alpha+D}}\rceil$ and $K(t)= t^{\frac{2\alpha}{3\alpha+D}}\log(t)$, we have
\begin{align*}
\mathbb{E}[R^\delta_\text{explore}(T)] \leq \frac{NB\lambda^{\max}d^{\max}}{\delta w^{\min}}2^D (T^{\frac{2\alpha+D}{3\alpha+D}}\log(T) +T^{\frac{D}{3\alpha+D}}) 
\end{align*}

To provide a upper bound of $\mathbb{E}[R^\delta_\text{exloit}(T)]$, we also need to identify suboptimal and near-optimal set of arms using:
\begin{align*}
\mathcal{L}^\delta(\p^t) = \left\{\f\in\mathcal{F} \mid U^t(\dot{\f}^\delta(\p^t); \ubar{\bm{\mu}}(\p^t)) - U^t(\f; \bar{\bm{\mu}}(\p^t)) \geq At^{\theta}\right\}
\end{align*}
where $\dot{\f}^\delta(\p^t)$ is the $\delta$-approximation solution to the following problem
\begin{align*}
\max\nolimits_{\f\in\mathcal{F}} ~ U^t(\f; \dot{\bm{\mu}}(\p^t)) \quad \text{s.t.} ~ \eqref{cstr:sub_c1},\eqref{cstr:sub_c2}, \eqref{cstr:sub_c3}
\end{align*}
and $\mathbb{E}[R^\delta_\text{exloit}(T)] = \mathbb{E}[R^\delta_\text{s}(T)] +\mathbb{E}[R^\delta_\text{n}(T)] $ is divided into two parts, which are bounded separately.

For the regret of choosing suboptimal decisions in $\mathcal{L}^\delta(\p^t)$ during exploitation, i.e., $\mathbb{E}[R^\delta_\text{s}(T)]$, we have  
\begin{align*}
\mathbb{E}[R^\delta_s(T)] & = \sum_{t=1}^T\sum_{\f\in\mathcal{L}^\delta(\p^t)} \mathbb{E}\left[I_{\{V_{f},W(t)\}}\times \left(\frac{1}{\delta}U^t(\f^{*,t};\bm{\lambda}^t)-U^t(\f;\bm{\lambda}^t)\right)\right]\\
 & \leq \frac{1}{\delta}\frac{B}{w^{\min}}\lambda^{\max}d^{\max}\sum_{t=1}^{T}\sum_{\f\in\mathcal{L}^\delta(\p^t)} \Pr\left\{V_{f}(t),W(t)\right\}
\end{align*}
Following similar steps in the proof of Lemma \ref{lemma:E_Rs}, it holds that 
\begin{align*}
\mathbb{E}[R^\delta_\text{s}(T)]\leq \frac{2|\mathcal{F}|B\lambda^{\max}d^{\max}}{\delta w^{\min}}\sum_{t=1}^T\sum_{n\in\mathcal{N}}\sigma_n\left(\frac{H(t)}{d^{\max}N},K(t)\right)
\end{align*}

For the regret of choosing near-optimal decisions in $\mathcal{F}\backslash\mathcal{L}^\delta(\p^t)$ during exploitation, i.e., $\mathbb{E}[R^\delta_\text{s}(T)]$, we have  
\begin{align*}
\mathbb{E}[R^\delta_\text{n}(T)] & = \sum_{t=1}^{T} \Pr\{W(t),Q(t)\} \cdot \mathbb{E}\left[\frac{1}{\delta}U^t(\f^{*,t};\bm{\lambda}^t)-U^t(\f^t;\bm{\lambda}^t) \mid  W(t), Q(t) \right]\\
& \leq  \sum_{t=1}^{T}\mathbb{E}\left[\frac{1}{\delta}U^t(\f^{*,t};\bm{\lambda}^t)-U^t(\f^t;\bm{\lambda}^t) \mid W(t),Q(t) \right].
\end{align*}
Since $\f^t \in \mathcal{F} \backslash \mathcal{L}^\delta(\p^t)$, it holds that 
\begin{align*}
U^t(\dot{\f}^\delta(\p^t); \ubar{\bm{\lambda}}(\p^t)) - U^t(\f^t; \bar{\bm{\lambda}}(\p^t)) < At^{\theta}.
\end{align*}
To bound the regret, we have to give an upper bound on
\begin{align*}
\sum_{t=1}^{T}\mathbb{E}\left[\frac{1}{\delta}U^t(\f^{*,t};\bm{\lambda}^t)-U^t(\f^t;\bm{\lambda}^t) \mid W(t), Q(t) \right] = \sum_{t=1}^{T} \frac{1}{\delta}U^t(\f^{*,t};\bm{\mu}^t)-U^t(\f^t;\bm{\mu}^t).
\end{align*}
Applying H\"{o}lder condition several times yields
\begin{align*}
\frac{1}{\delta}U^t(\f^{*,t};\bm{\mu}^t)-U^t(\f^t;\bm{\mu}^t) \leq & 	\frac{1}{\delta}U^t(\f^{*,t};\dot{\bm{\mu}}(\p^t)) + d^{\max}NLD^{\frac{\alpha}{2}}h_{T}^{-\alpha} - U^t(\f^t;\bm{\mu}^t) \\
\leq & \frac{1}{\delta} U^t(\dot{\f}(\p^t);\dot{\bm{\mu}}(\p^t)) + d^{\max}NLD^{\frac{\alpha}{2}}h_{T}^{-\alpha} - U^t(\f^t;\bm{\mu}^t) \\
\leq &  U^t(\dot{\f}^\delta(\p^t);\dot{\bm{\mu}}(\p^t)) + d^{\max}NLD^{\frac{\alpha}{2}}h_{T}^{-\alpha} - U^t(\f^t;\bm{\mu}^t) \\
\leq & 	U^t(\dot{\f}^\delta(\p^t);\ubar{\bm{\mu}}(\p^t)) + 2d^{\max}NLD^{\frac{\alpha}{2}}h_{T}^{-\alpha} - U^t(\f^t;\bm{\mu}^t) \\
\leq & 	U^t(\dot{\f}^\delta(\p^t);\ubar{\bm{\mu}}(\p^t)) + 3d^{\max}NLD^{\frac{\alpha}{2}}h_{T}^{-\alpha} - U^t(\f^t;\bar{\bm{\mu}}(\p^t))\\
\leq & 3d^{\max}NLD^{\frac{\alpha}{2}}h_{T}^{-\alpha}  + At^\theta
\end{align*}
and therefore
\begin{align*}
\mathbb{E}[R^\delta_\text{exploit}(T)] \leq \frac{2|\mathcal{F}|B\lambda^{\max}d^{\max}}{\delta w^{\min}}\sum_{t=1}^T\sum_{n\in\mathcal{N}}\sigma_n\left(\frac{H(t)}{d^{\max}N},K(t)\right) + 3d^{\max}NLD^{\frac{\alpha}{2}}h_{T}^{-\alpha}T + AT^{\theta+1}.
\end{align*}
Letting $h_T = \lceil T^{\frac{1}{3\alpha+D}}\rceil$ and $K(t)= t^{\frac{2\alpha}{3\alpha+D}}\log(t)$, and assuming that MLE is applied for service demand estimation, we have
\begin{align*}
\mathbb{E}[R^\delta_\text{exploit}(T)] \leq  \frac{N|\mathcal{F}|B\lambda^{\max}d^{\max}\pi^2}{3\delta w^{\min}} +  3d^{\max}NLD^{\frac{\alpha}{2}}T^{\frac{2\alpha+D}{3\alpha+D}} + AT^{\frac{2\alpha+D}{3\alpha+D}}
\end{align*}
The proof is completed by comparing the leading orders of the upper bounds of $\mathbb{E}[R^\delta_\text{explore}(T)]$ and $\mathbb{E}[R^\delta_\text{exploit}(T)]$.
\end{proof}

\end{document}